\documentclass[a4paper]{article}
\usepackage{hyperref}
\usepackage{bbding}
\usepackage[T1]{fontenc}
\usepackage[utf8]{inputenc}
\usepackage{aurical}
\usepackage{huncial}
\usepackage{linearb}
\usepackage{textcomp}
\usepackage{amsmath}
\usepackage{amssymb}
\usepackage[amsmath,amsthm,thmmarks]{ntheorem}
\usepackage{amsfonts}
\usepackage{amssymb,epic,eepic}
\usepackage{stmaryrd}
\usepackage{capt-of}
\usepackage{graphicx}
\usepackage[usenames,dvipsnames]{pstricks}
\usepackage{epsfig}
\usepackage{pst-grad}
 \usepackage{pst-plot}
\usepackage{here}
\usepackage[justification=centering]{caption}
\newtheorem{theorem}{Theorem}[section]

\newtheorem{definition}[theorem]{Definition}
\newtheorem{remark}[theorem]{Remark}

\newtheorem{example}[theorem]{Example}

\begin{document}
\title{Classical-Quantum Arbitrarily Varying Wiretap Channel: Ahlswede Dichotomy, Positivity,
Resources, Super Activation}
\author{Holger Boche\\
Lehrstuhl f\"ur Theoretische
Informationstechnik,\\
 Technische Universit\"at M\"unchen,\\
Munich, Germany\\
boche@tum.de \and Minglai Cai\\
Lehrstuhl f\"ur Theoretische
Informationstechnik,\\
 Technische Universit\"at M\"unchen,\\
Munich, Germany\\
minglai.cai@tum.de \and Christian Deppe \\
Fakult\"at f\"ur
 Mathematik\\ 
Universit\"at Bielefeld,\\
Bielefeld, Germany,\\
cdeppe@mathematik.uni-bielefeld.de  \and Janis N\"otzel  \\
 Universitat Aut\`{o}noma de Barcelona,\\
 Barcelona, Spain\\
Janis.Notzel@uab.cat}

%\titlerunning{Classical-Quantum Arbitrarily Varying Wiretap Channel}
\maketitle
\begin{abstract}
We establish the Ahlswede dichotomy for arbitrarily varying
classical-quantum
wiretap channels, i.e., either the deterministic secrecy capacity of
the channel is zero, or it equals its randomness-assisted
secrecy capacity.
We analyze the secrecy capacity of these channels when the sender and the
receiver use various resources. It turns out that randomness,
common randomness, and correlation as resources are very helpful for
achieving a positive secrecy capacity.
We prove the phenomenon
``super-activation'' for arbitrarily varying classical-quantum
wiretap channels, i.e., two channels, both with zero
deterministic secrecy capacity, if used together allow perfect
secure transmission. 
\end{abstract}

\tableofcontents
\section{Introduction}\label{tdimcsareqcsw}

The developments in modern communication systems are rapid. 
Especially quantum communication systems allow us to exploit 
new possibilities while at the same time imposing fundamental 
limitations. Quantum information processing systems provide huge 
theoretical advantages over their classical counterparts, one of the two most prominent ones being \emph{perfect secrecy}  (cf.
\cite{Be/Br} and \cite{Be} for two well-known examples of 
quantum key distributions). The impact of quantum information 
processing systems on our daily live is nonetheless still zero, 
the main reason for that being the difficulty to store and manipulate quantum states in a predictable and reliable manner.

In this work, we bring these two aspects together, namely we 
investigate the transmission of messages from a sending to a 
receiving party. The messages ought to be kept secret from an 
eavesdropper. Communication takes place over a quantum channel 
which is, in addition to noise from the environment, subjected 
to the action of a jammer which actively manipulates the states.

Preceding work in quantum information theory  has mostly focused on either 
of the two attacks. Our goal is to deliver a more general theory considering 
both channel robustness and security in quantum information theory. By doing 
so, we build on the preceding works \cite{Bl/Ca} and \cite{Bo/No}. 
Furthermore, we are interested in the delivery of large volumes of 
messages over many channel uses, so that we study the asymptotic behavior of the system.

Our work fits into a broader range of recent work in both classical and quantum information theory 
\cite{Wi/No/Bo,Ca/Wi/Ye,Bo/Ca/Ca/De,Bl/Ca,Bj/Bo/So,Bo/Wy,De,he-khisti-yener,ITW2010,Li/Kr/Po/Sh} that
 studies the secret information processing tasks with the aim of delivering \emph{embedded security}: Unlike what is nowadays the standard approach in secret communication, 
namely to first ensure the successful transmission of messages and then implement a cryptographic protocol on top whose security relies on assumptions concerning the difficulties 
in breaking the protocol, this new paradigm focuses on delivering a guaranteed security right from 
the start. The security features of the protocol become embedded already at the physical 
layer of the communication system. The concept does not only cover secure message transmission but also secure key generation.

Also, communication models including a jammer that tries to prevent the legal parties 
from communicating properly have received a great lot of attention in recent years, 
some of which we have already mentioned above include wiretapping aspects, while some do not \cite{Ahl/Bj/Bo/No}.
These publications concentrated on the model of an \emph{arbitrarily varying channel} 
where the jammer may change his input in every channel use and is not restricted to use a 
repetitive probabilistic strategy. Quite on the contrary, it is understood that the sender 
and the receiver have to select their coding scheme first. After that the jammer makes his 
choice of the channel state. The model of an arbitrarily varying channel was first introduced by Blackwell, Breiman, and Thomasian in \cite{Bl/Br/Th2}.
The nature of the model is quite flexible: It allows specifying the impact that the actions of the jammer may have on the communication link under 
use: In the most restrictive case where the jammer is left with only one choice, we recover the discrete memoryless channel. On the other extreme, it 
has been shown by Ahlswede in \cite{Ahl0}  that the capacity (under maximal error criterion) of certain arbitrarily varying channels can be equated to the 
zero-error capacity of related discrete memoryless channels.
The arbitrarily varying channel does at the same time demonstrate the importance of shared randomness for communication in a very clear form: Ahlswede showed 
in \cite{Ahl1}  (cf. also \cite{Ahl2} and \cite{Ahl3}) the surprising result that 
either the deterministic capacity of an arbitrarily varying channel is zero, 
or it equals its shared randomness-assisted capacity (this effect is now known as 
the Ahlswede dichotomy). After that discovery, it remained an open question exactly 
when the deterministic capacity is nonzero. In \cite{Rei} Ericson gave a sufficient 
condition for that, and in \cite{Cs/Na} Csisz\'ar and Narayan proved that this is condition is also necessary.

In this work, we will therefore put a focus on the analysis of different 
forms of shared randomness and their impact on the robustness and security.
The model of a wiretap channel adds a third party to the communication 
problem as well, but here the focus is on secure communication, meaning 
communication without that third party getting to know the messages. This 
model was first introduced by Wyner in \cite{Wyn} (in this paper we will 
use a stronger security criterion than the one that was used in \cite{Wyn}, cf. Remark \ref{remsc}).
The relation of the different security criteria is discussed, e.g. in 
\cite{Bl/La} with some generality and in \cite{Wi/No/Bo} with respect to arbitrarily varying channels.

In the model of an arbitrarily varying wiretap channel, we consider transmission 
with both a jammer and an eavesdropper. Its secrecy capacity has been analyzed in 
\cite{Bj/Bo/So}. A lower bound of the randomness-assisted secrecy capacity has been 
given. It is worth noting that the channel under consideration in this work is 
effectively given by an \emph{interference channel} where the legal sender and 
the jammer are allowed to make inputs to the system and the legal receiver as 
well as the eavesdropper receive the corresponding outputs. We do leave open
 the possibility of the jammer communicating his choice of input (equivalently: 
his channel state sequence) to the eavesdropper, but limit the receiving parties 
such that they cannot send any messages back to the jammer or the legal receiver.
During proofs and when defining the model, we will, however deviate from this point 
of view and use a notation which respects the historic development of results on 
arbitrarily varying channels.
The physical model we consider is that of a classical-quantum channel, i.e., 
the legal sender's and the jammer's inputs are classical data and the legal 
receiver's as well as the eavesdroppers outputs are quantum systems. The capacity 
of classical-quantum channels without secrecy constraints or active jamming has 
been determined in \cite{Ho} and \cite{Sch/Wes}.

A classical-quantum channel with a jammer is called an arbitrarily varying 
classical-quantum channel. In \cite{Ahl/Bli} the capacity of arbitrarily varying 
classical-quantum channels is analyzed. A lower bound of the capacity has been given.
 An alternative proof  and a proof of the strong converse are given 
in \cite{Bj/Bo/Ja/No}. In \cite{Ahl/Bj/Bo/No} the Ahlswede dichotomy for the arbitrarily varying 
classical-quantum channels is established, and a sufficient and necessary condition for the zero 
deterministic capacity is given. In \cite{Bo/No} a simplification of this condition for the 
arbitrarily varying classical-quantum channels is given. A classical-quantum 
channel with an eavesdropper is called a classical-quantum wiretap channel, 
its secrecy capacity has been determined in \cite{De} and \cite{Ca/Wi/Ye}.

 A classical-quantum channel with both a jammer and an eavesdropper is called 
an arbitrarily varying classical-quantum wiretap channel. It is defined as a 
family of pairs of indexed channels $\{(W_t,V_t):t=1,\ldots,T\}$ with a common 
input alphabet and possible different output systems, connecting a sender with two 
receivers, a legal one and a wiretapper, where $t$ is called a channel state of the 
channel pair. The legitimate receiver accesses the output of the first part of the pair, 
i.e., the first channel $W_t$ in the pair, and the wiretapper observes the output of the 
second part, i.e., the second channel $V_t$, respectively. A channel state $t$, 
which varies from symbol to symbol in an arbitrary manner, governs both the 
legal receiver's channel and the wiretap channel. A code for the channel 
conveys information to the legal receiver such that the wiretapper knows 
nothing about the transmitted information in the sense of the stronger 
security criterion (cf. Remark \ref{remsc}). This is a generalization of 
compound classical-quantum wiretap channels in \cite{Bo/Ca/Ca/De}, when the channel states are not stationary, but can change over time.

The secrecy capacity of the arbitrarily varying classical-quantum wiretap channels has been analyzed in \cite{Bl/Ca}. A lower 
bound of the randomness-assisted capacity has been given, and it has been shown that this 
bound is either a lower bound for the deterministic capacity, or else the deterministic capacity is equal to zero.
As mentioned already, we will be interested in the role that different forms of shared randomness play for the arbitrarily 
varying classical-quantum wiretap channel.
To this end, we will distinguish between three  kinds of shared randomness: \emph{randomness}, \emph{common randomness}, and \emph{correlation}.
Randomness and common randomness have been used as a method of proof, e.g., in \cite{Ahl1} and much of the follow-up work for the determination of 
the random capacity. If looked at as a resource for communication which is to be deployed in order to make a communication 
link work reliably, they are, however, a rather strong form of a resource: It is required that both sender and receiver 
have access to a perfect copy of the outcome of a random experiment. Moreover, the outcomes of said experiment 
have to be distributed uniformly. The impact of deviations from these strong requirements has not yet received 
much attention. What has been investigated (starting with \cite{Ahl/Cai} and continued in \cite{Bo/No}) is a 
variant where the common randomness gets replaced by a resource that is in some sense the complete opposite: correlation.

Assume that a bipartite source, modeled by an i.i.d. random variable $(X,Y)$  
with values in a finite product set $\mathbf X\times\mathbf Y$, is observed by the 
sender and (legal) receiver. The sender has access to the random variable $X$ and 
the receiver to $Y$. We call $(X,Y)$ correlated shared randomness whenever the 
mutual information between $X$ and $Y$ satisfies $I(X;Y)>0$.

It has been shown in \cite{Ahl/Cai} that correlated shared randomness is a helpful resource for information transmission through an arbitrarily varying classical channel: 
The use of mere correlation does already allow one to transmit messages at any rate that would be achievable using any form of shared randomness. 
The capacity of an arbitrarily varying quantum channel assisted by correlated shared randomness as resource has been discussed in \cite{Bo/No},
 where equivalent results were found. 
In this work, we extend the concept of correlation-assisted coding to the arbitrarily varying classical-quantum wiretap channel.

In \cite{Bo/No} a classification of various resources is given. A distinction is made between  two extremal cases:
randomness and correlation. Randomness is the strongest resource, it
requires a perfect copy of the
outcome of a random experiment, and thus we should assume
an additional perfect channel.
On the other hand, correlation is the weakest resource.  
The work \cite{Bo/No} also put emphasis on the quantification of the differences between 
correlation and common randomness and used the arbitrarily varying classical-quantum channel as a method of proof.
It can be shown that common randomness is a stronger resource than  correlation in the following sense:
An example is given when not even a finite amount of common randomness can be extracted from a given  correlation.
On the contrary,  a sufficiently large amount of common randomness  allows the sender and receiver to 
asymptotically simulate the statistics of any correlation.

We concentrate our analysis on the case without feedback, i.e., we neither allow the receiver to send messages back to the sender (or the jammer), 
nor do we allow the eavesdropper to send messages toward the jammer (or the sender).
Such an approach may be deemed unsatisfactory from a practical perspective. However, a brief look into the history of the arbitrarily varying channel 
reveals that only the reduction to the case of deterministic codes (without feedback) leads one to encounter those cases where the capacity of the system 
is zero, while a dramatic increase to full capacity is possible as soon as shared randomness (or feedback) is available.

In the situation investigated here, the reduction to forward communication allows us to demonstrate the effect of super-activation of the secrecy capacity 
of the arbitrarily varying classical-quantum channel. We take the space to write a few lines concerning more elaborate models. The case where a (secure) 
channel from the receiver to the sender is available is likely to be equivalent to the case where shared randomness can be used when the average error 
criterion is used. The latter will be treated in forthcoming work. In the model treated here, the presence of an eavesdropper makes us take the freedom to 
allow randomness at the encoder, which makes the average error criterion equivalent to the maximal error criterion \cite{Ahl1} (and \cite{Bo/No} for the 
quantum case). The case of deterministic codes in the presence of feedback but with the code performance being evaluated with respect to the maximal 
error criterion has been evaluated in \cite{Ahl-feedback}.

It can easily be seen now that the complexity of the channel model under investigation here necessitates a 
strict reduction in the abilities of the participating parties, at least if the aim is the establishment of definite results.

Our secrecy criterion is chosen such that the messages sent by the sender are to be kept \emph{strongly} 
secret. More precisely, the use of shared randomness creates ensembles 
$(R_{\mathrm{uni}},Z_{t^n}):=(J^{-1},V_{t^n}(E^\gamma(\cdot|j))_{j=1}^J$ where $j=1,\ldots,J$ are messages 
and $E^\gamma(\cdot|1),\ldots,E^\gamma(\cdot|J)$ are probability distributions of the codewords associated with the messages. 
The index $\gamma$ refers to a particular choice of encoding scheme. This index may be shared with the receiver 
(common randomness) or may just be correlated with another index $\gamma'$ at the receiver (correlated codes, 
in that case both $\gamma$ and $\gamma'$ are actually elements of product alphabets $\mathbf X^n$ and $\mathbf Y^n$).

Our strong secrecy criterion requires that the Holevo information $\chi(R_{\mathrm{uni}},Z_{t^n,\gamma})$ of the 
ensembles consisting of the messages and the output at the eavesdropper's system is to be kept small in a 
yet to be defined sense. More precisely, we require that the Holevo information is to be kept small on 
average over the random choice of codewords and for all possible choices of the jammer, i.e., 
$\max_{t^n}\chi(R_{\mathrm{uni}},\mathcal Z_{t^n}|\Gamma)$ is to vanish asymptotically.

Using this secrecy criterion is a key to prove super-activation of the 
deterministic secrecy capacity of the arbitrarily varying classical-quantum 
channel: We take two arbitrarily varying classical-quantum wiretap channels. 
One of them is assumed to have zero capacity for message transmission because 
it is symmetrizable in the sense of \cite{Ahl/Bli}, but its common-randomness-assisted
 capacity is positive. The other is assumed to be non-symmetrizable but insecure.

In 
\cite{Wi/No/Bo2}
a new code concept for secrecy capacity 
 and a complete characterization of super-activation
for classical
arbitrarily varying wiretap channels 
with no sharing resources
has been given.
In view
of this work on  classical
arbitrarily varying wiretap channels our further
task will be to analyze this characterization on
arbitrarily varying classical-quantum wiretap channels.

Through parallel transmission of common randomness on the insecure 
channel and secure data on the other one, the combined system 
can be proven to have positive capacity. Roughly speaking, the 
proof uses the fact that the choices of common randomness and 
messages are independent from each other and only the codewords 
depend on both of them, together with the data processing inequality 
applied to the Holevo quantity. Details are to be found in the respective section.

The operational interpretation of the secrecy criterion that we 
employ here comes through application of the (quantum) Pinsker's 
inequality. Note that, in an average sense, the eavesdropper ``knows'' 
the index $\gamma$ of the random code. It is clear that, under such 
circumstances, backwards communication toward the jammer would render the shared randomness completely useless.

A more in-depth discussion of secrecy criteria in the quantum case, 
including fully quantum channels but not the arbitrarily varying case, 
can be found in the recent preprint \cite{winter-locking}. Different 
secrecy criteria for arbitrarily varying quantum or classical-quantum channels will be evaluated in future work.

 \vspace{0.15cm}

 This
paper is organized as follows.\vspace{0.15cm}

The main definitions  are given in
Section \ref{BDaCS}.

In Section  \ref{ADFAVCQWC} we generalize the result of \cite{Bl/Ca} by  establishing
the Ahlswede dichotomy for the  arbitrarily varying classical-quantum  wiretap  channels (without feedback),
i.e.,  either the deterministic secrecy  capacity of an
arbitrarily varying classical-quantum
wiretap  channel is zero, or it equals its randomness-assisted
 secrecy  capacity.

In Section  \ref{AVCQWCWCA} we   analyze the secrecy capacity
of an
 arbitrarily varying classical-quantum  wiretap  channel assisted by 
correlation as resource. We show that 
correlation  is  a helpful resource for  secure information
transmission through an arbitrarily varying classical-quantum
wiretap    channel.

In Section \ref{AP} we give an example in which both cases of the
 Ahlswede dichotomy for the  arbitrarily varying classical-quantum  wiretap  channels
  actually occur.
We  present a new discovery for the
 arbitrarily varying classical-quantum  wiretap  channels
which is a consequence of the Ahlswede dichotomy for the  arbitrarily varying
classical-quantum  wiretap  channels. This  phenomenon is called
``super-activation'', i.e., two arbitrarily varying classical-quantum
 wiretap channels, both with zero   deterministic  secrecy  capacity,
if used  together allow perfect secure transmission. 
%We will also  discuss the
%difficulty in finding a similar result to the  Ahlswede dichotomy for
%the entanglement generating capacity of an arbitrarily varying
%quantum
% channel
%(cf. also Conjecture \ref{AhlswedeConjecture} and the discussion in
%Section \ref{BDaCS}).

Finally, we will conclude in Section \ref{concl} with a discussion of our
results.

\section{Communication
Scenarios and Resources}\label{BDaCS}

\subsection{Basic Definitions and Communication Scenarios}

\allowdisplaybreaks 

For  a finite set $\mathbf{A}$, we denote  the
set of probability distributions on $\mathbf{A}$ by $P(\mathbf{A})$.
Let $H$ be a finite-dimensional
complex Hilbert space. We denote the (convex) space of  density operators on $H$
by  $\mathcal{S}(H)$. A classical-quantum channel is
a linear map $W: P(\mathbf{A})\rightarrow\mathcal{S}(H)$,
$P(\mathbf{A})  \ni P
\rightarrow W(P) \in \mathcal{S}(H)$.
Let $a\in \mathbf{A}$. For a $P_a\in P(\mathbf{A})$, 
defined by $P_a(a')= \begin{cases} 1 &\mbox{if } a'=a\\
0 &\mbox{if } a'\not=a \end{cases}$, we 
write $W(a)$ instead of $W(P_a)$.

\begin{remark} In many literature, a classical-quantum channel is
defined as a map $ \mathbf{A}\rightarrow\mathcal{S}(H)$,
$\mathbf{A}  \ni a
\rightarrow W(a) \in \mathcal{S}(H)$. This is a 
special case when the input is limited on the set $\{P_a :
a\in \mathbf{A}\}$. \end{remark}

For any finite set $\mathbf{A}$, any finite-dimensional
complex Hilbert space $H$, and  $n\in\mathbb{N}$, we define ${\mathbf{A}}^n:= \Bigl\{(a_1,\ldots,a_n): a_i \in \mathbf{A}
\text{ } \forall i \in \{1,\ldots,n\}\Bigr\}$, and $H^{\otimes n}:=
span\Bigl\{v_1\otimes \ldots \otimes v_n: v_i \in H
\text{ } \forall i \in \{1,\ldots,n\}\Bigr\}$. We also write $a^n$ 
for the elements of
${\mathbf{A}}^n$.

Associated with $W$ is the channel map on the n-block $W^{\otimes n}$: $P({\mathbf{A}}^n)
\rightarrow \mathcal{S}({H}^{\otimes n})$, such that 
 $W^{\otimes n}(P^n) = W(P_1)
\otimes \ldots \otimes W(P_n)$
if $P^n\in P(\mathbf{A}^{ n})$ can be written as
$(P_1,\ldots,P_n)$. 
Let $\theta$ $:=$ $\{1,\ldots,T\}$ be a finite set.
Let $\Bigl\{W_t:t\in\theta\Bigr\}$ be a set of classical-quantum channels.
For $t^n=(t_1,\ldots,t_n)$, $t_i\in\theta$ we define the n-block $W_{t^n}$ 
such that  for
 $W_{t^n}(P^n) = W_{t_1}(P_1)
\otimes \ldots \otimes W_{t_n}(P_n)$
if $P^n\in P(\mathbf{A}^{ n})$ can be written as
$(P_1,\ldots,P_n)$.  
\vspace{0.15cm}

Let $\mathfrak{P}$ and $\mathfrak{Q}$ be
quantum systems, denote the Hilbert space of $\mathfrak{P}$ and
$\mathfrak{Q}$ by $H^\mathfrak{P}$ and $H^\mathfrak{Q}$,
respectively.
We denote the space of  density operators on $H^\mathfrak{P}$ and  $H^\mathfrak{Q}$
by  $\mathcal{S}(H^\mathfrak{P})$ and $\mathcal{S}(H^\mathfrak{Q})$, respectively. A
quantum channel $N$: $\mathcal{S}(H^\mathfrak{P}) \rightarrow \mathcal{S}(H^\mathfrak{Q})$,
$\mathcal{S}(H^\mathfrak{P})  \ni \rho
\rightarrow N(\rho) \in \mathcal{S}(H^\mathfrak{Q})$
is represented by a completely positive trace preserving
map,
 which accepts input quantum states in $\mathcal{S}(H^\mathfrak{P})$ and produces output quantum
states in  $\mathcal{S}(H^\mathfrak{Q})$.

Associated with $N$ is
the channel maps on the n-block ${N}^{\otimes n}$: $\mathcal{S}({H^\mathfrak{P}}^{\otimes n}) \rightarrow
\mathcal{S}({H^\mathfrak{Q}}^{\otimes n})$ such that for $\rho^{n} = \rho_1\otimes \ldots \otimes \rho_n
\in \mathcal{S}({H^\mathfrak{P}}^{\otimes n})$
 ${N}^{\otimes n}(\rho^{n}) = {N}(\rho_1)
\otimes \ldots \otimes {N}(\rho_n)$.
For $t^n=(t_1,\ldots,t_n)$, $t_i\in\theta$, we define the $n$-block $N_{t^n}$ 
such that  for $\rho^{n} = \rho_1\otimes \ldots \otimes \rho_n
\in \mathcal{S}({H^\mathfrak{P}}^{\otimes n})$ we have
 $N_{t^n}(\rho^n) = N_{t_1}(\rho_1)
\otimes \ldots \otimes N_{t_n}(\rho_n)$.

We denote the identity operator on a space $H$ by $\mathrm{id}_H$.\vspace{0.15cm}

For a discrete random variable $X$  on a finite set $\mathbf{A}$ and a discrete
random variable  $Y$  on  a finite set $\mathbf{B}$,  we denote the Shannon entropy
of $X$ by
$H(X)=-\sum_{x \in \mathbf{A}}p(x)\log p(x)$ and the mutual information between $X$
and $Y$ by  
$I(X;Y) = \sum_{x \in \mathbf{A}}\sum_{y \in \mathbf{B}}  p(x,y) \log{ \left(\frac{p(x,y)}{p(x)p(y)} \right) }$.
Here $p(x,y)$ is the joint probability distribution function of $X$ and $Y$, and 
$p(x)$ and $p(y)$ are the marginal probability distribution functions of $X$ and $Y$, respectively,
and ``$\log$''  means logarithm to base $2$.\vspace{0.15cm}

For a quantum state $\rho\in \mathcal{S}(H)$, we denote the von Neumann
entropy of $\rho$ by \[S(\rho)=- \mathrm{tr}(\rho\log\rho)\text{
,}\] where ``$\log$''  means logarithm to base $2$.
Let $\Phi := \{\rho_x : x\in \mathbf{A}\}$  be a set of quantum  states
labeled by elements of $\mathbf{A}$. For a probability distribution  $Q$
on $\mathbf{A}$, the    Holevo $\chi$ quantity is defined as
\[\chi(Q;\Phi):= S\left(\sum_{x\in \mathbf{A}} Q(x)\rho_x\right)-
\sum_{x\in \mathbf{A}} Q(x)S\left(\rho_x\right)\text{ .}\]
Note that we can always associate a state 
$\rho^{XY}=\sum_{x}Q(x)|x\rangle\langle x|\otimes \rho_x$ to
$(Q;\Phi)$ such that $\chi(Q;\Phi)=I(X;Y)$ holds for the quantum
mutual information.
\vspace{0.15cm}

\begin{definition}
Let $\mathbf{A}$ be a finite set. Let
 $H$  be a finite-dimensional
complex Hilbert space, and
  $\theta$ $:=$ $\{1,\ldots,T\}$ be a finite set.
    For every $t \in \theta$,  let $W_{t}$   be a classical-quantum channel
$P(\mathbf{A}) \rightarrow \mathcal{S}(H)$.
The set of the  quantum
channels  $\{W_t : t \in \theta\}$
defines an  arbitrarily varying
classical-quantum  channel.
\end{definition}

Strictly speaking, the set $\{W_t : t \in \theta\}$ generates the 
arbitrarily varying
classical-quantum  channel $\{W_{t^n} : t^n \in \theta^n\}$.
When the sender inputs a  $P^n \in P({\mathbf{A}}^n)$ into the channel, the receiver
receives the output $W_{t^n}(P^n)
 \in \mathcal{S}(H^{\otimes n})$, where $t^n = (t_1, t_2, \ldots ,
t_n)\in\theta^n$ is the channel state of $W_{t^n}$.

\begin{definition}\label{symmet}
We say that the arbitrarily varying classical-quantum  channel
$\{W_t : t \in \theta\}$ is symmetrizable if
 there exists a
parametrized set of distributions $\{\tau(\cdot\mid a):
 a\in \mathbf{A}\}$ on $\theta$ such that for all $a$, ${a'}\in \mathbf{A}$,
\[\sum_{t\in\theta}\tau(t\mid a)W_{t}({a'})=\sum_{t\in\theta}\tau(t\mid {a'})W_{t}(a)\text{ .}\]
\end{definition}

\begin{definition}
Let $\mathfrak{P}$ and $\mathfrak{Q}$ be
quantum systems, denote the Hilbert Space of $\mathfrak{P}$ and
$\mathfrak{Q}$ by $H^\mathfrak{P}$ and $H^\mathfrak{Q}$,
respectively, and let
  $\theta$ $:=$ $\{1,\ldots,T\}$ be a finite set. For every $t \in \theta$,  let ${W'}_{t}$   be a quantum channel
$\mathcal{S}(H^\mathfrak{P}) \rightarrow \mathcal{S}(H^\mathfrak{Q})$.
We call the set of the  quantum
channels  $\{{W'}_t: t \in \theta\}$ an  arbitrarily varying
quantum  channel
when the state $t$ varies from
symbol to symbol in an  arbitrary manner.
We denote the set of  arbitrarily varying
quantum channels
$\mathcal{S}(H^\mathfrak{P}) \rightarrow \mathcal{S}(H^\mathfrak{Q})$ 
by $C(H^\mathfrak{P},H^\mathfrak{Q})$.
\end{definition}

\begin{definition}
Let $\mathbf{A}$ be a finite set. Let
 $H$ and $H'$ be finite-dimensional
complex Hilbert spaces. Let
  $\theta$ $:=$ $\{1,\ldots,T\}$ be a finite set. For every $t \in \theta$  let $W_{t}$  
	be a classical-quantum channel
$P(\mathbf{A}) \rightarrow \mathcal{S}(H)$ and ${V}_t$ be a classical-quantum channel $P(\mathbf{A})
\rightarrow \mathcal{S}(H')$. We call the set of the  classical-quantum
channel pairs  $\{(W_t,{V}_t): t \in \theta\}$ an \bf arbitrarily
varying classical-quantum wiretap channel\it, the legitimate
receiver accesses the output of the first channel, i.e., $W_t$  in the pair
$(W_t,{V}_t)$, and the wiretapper observes the output of the second
 channel, i.e.,  ${V}_t$ in the pair $(W_t,{V}_t)$, respectively, when the
state $t$ varies from symbol to symbol in an  arbitrary
manner.\end{definition}

When the sender inputs  a sequence $a^n \in {\mathbf{A}}^n$  into the channel, the receiver
receives the output $W_{t^n}(a^n)
 \in \mathcal{S}(H^{\otimes n})$, where $t^n = (t_1, t_2, \ldots ,
t_n)\in\theta^n$ is the channel state, while the wiretapper  receives an output quantum  state ${V}_{t^n}(a^n)
 \in \mathcal{S}({H}'^{\otimes n})$.\vspace{0.2cm}

\subsection{Code Concepts and Resources}

Our goal is to see what
 the effects on the secrecy capacities of an arbitrarily
varying classical-quantum wiretap channel are if  the sender and the legal
receiver have the possibility to use various kinds of resources.
 We also want to
investigate  what amount of randomness   is necessary for the robust and
secure message transmission through an arbitrarily varying
classical-quantum wiretap channel.
 Hence, we consider
various kinds of resources, each of them requiring a different amount of randomness, and
we consider different codes, each of
them requiring a different kind of resource.

\begin{definition}
 An $(n, J_n)$   (deterministic)  code $\mathcal{C}$ for the
arbitrarily
varying classical-quantum wiretap channel $\{(W_t,{V}_t): t \in \theta\}$
consists of a stochastic encoder $E$ : $\{
1,\ldots ,J_n\} \rightarrow P({\mathbf{A}}^n)$, 
$j\rightarrow E(\cdot|j)$,
 specified by
a matrix of conditional probabilities $E(\cdot|\cdot)$, and
 a collection of positive semi-definite operators $\left\{D_j: j\in \{ 1,\ldots ,J_n\}\right\}$
on ${H}^{\otimes n}$,
which is a partition of the identity, i.e., $\sum_{j=1}^{J_n} D_j =
\mathrm{id}_{{H}^{\otimes n}}$. We call these  operators the decoder operators.
\end{definition}
A code is created by the sender and the legal receiver before the
message transmission starts. The sender uses the encoder to encode the
message that he wants to send, while the legal receiver uses the
decoder operators on the channel output to decode the message.

\begin{remark}
 An $(n, J_n)$   deterministic code $\mathcal{C}$ with deterministic encoder
consists of a family of $n$-length strings of symbols $\left(c_j\right)_{j\in
\{1,\ldots ,J_n\}} \in \left({\mathbf{A}}^n\right)^{J_n}$   and
 a collection of positive semi-definite operators $\left\{D_j: j\in \{ 1,\ldots ,J_n\}\right\}$
on ${H}^{\otimes n}$
which is a partition of the identity.

 The deterministic encoder is a special case of the
 stochastic encoder when we require that for every $j\in
\{1,\ldots ,J_n\}$, there is a sequence $a^n\in {\mathbf{A}}^n$  chosen with probability $1$.
 The standard technique for  message transmission over a  channel and
robust message transmission over an arbitrarily
varying  channel
 is to use the deterministic encoder
(cf.  \cite{Ahl/Bj/Bo/No} and \cite{Bo/No}).
 However, we use  the stochastic encoder, since it is  a tool
 for  secure message transmission over wiretap channels (cf. \cite{Bl/Ca} and \cite{Ahl/Bli}).
\label{detvsran}
\end{remark}

\begin{definition}
A nonnegative number $R$ is an achievable  (deterministic) secrecy
rate for the arbitrarily varying classical-quantum wiretap channel
$\{(W_t,{V}_t): t \in \theta\}$ if for every $\epsilon>0$, $\delta>0$,
$\zeta>0$ and sufficiently large $n$ there exist an  $(n, J_n)$
code $\mathcal{C} = \bigl(E, \{D_j^n : j = 1,\ldots J_n\}\bigr)$  such that $\frac{\log
J_n}{n}
> R-\delta$, and
\begin{equation} \max_{t^n \in \theta^n} P_e(\mathcal{C}, t^n) < \epsilon\text{ ,}\label{annian1}\end{equation}
\begin{equation}\max_{t^n\in\theta^n}
\chi\left(R_{\mathrm{uni}};Z_{t^n}\right) < \zeta\text{
,}\label{b40}\end{equation} where $R_{\mathrm{uni}}$ is the uniform
distribution on $\{1,\ldots J_n\}$. Here
$P_e(\mathcal{C}, t^n)$ (the average probability of the decoding error of a
deterministic code $\mathcal{C}$, when the channel state  of the
arbitrarily varying classical-quantum wiretap channel $\{(W_t,{V}_t): t \in \theta\}$  is $t^n =
(t_1, t_2, \ldots , t_n)$), is defined as \[ P_e(\mathcal{C}, t^n) := 1- \frac{1}{J_n} \sum_{j=1}^{J_n}
\mathrm{tr}(W_{t^n}(E(~|j))D_j)\text{ ,}\] 
$Z_{t^n}=\Bigl\{{V}_{t^n}(E(~|i)):$ 
$i\in\{1,\ldots,J_n\}\Bigr\}$ 
 is the set of the resulting
quantum state at the output of the wiretap channel when the channel
state of $\{(W_t,{V}_t): t \in \theta\}$ is $t^n$.
\end{definition}

\begin{remark} A weaker and widely used
 security criterion is obtained if we replace
(\ref{b40})  with $\max_{t \in \theta}\frac{1}{n}
\chi\left(R_{\mathrm{uni}};Z_{t^n}\right)$ $<$ $\zeta$. In this paper  we
will follow
  \cite{Bj/Bo/So} and use  (\ref{b40}).
\label{remsc}\end{remark}

\begin{remark}When we defined  $W_{t}$ as
$ \mathbf{A}\rightarrow\mathcal{S}(H)$, then
$P_e(\mathcal{C}, t^n)$ is defined as $1-$ $\frac{1}{J_n} \sum_{j=1}^{J_n}  \sum_{a^n \in {\mathbf{A}}^n}$
$E(a^n|j)\mathrm{tr}(W_{t^n}(a^n)D_j)$.

When deterministic encoder is used, then
$P_e(\mathcal{C}, t^n)$ is defined as $1-$ $\frac{1}{J_n} \sum_{j=1}^{J_n} $
$\mathrm{tr}(W_{t^n}(c_j)D_j)$.
\end{remark}\vspace{0.15cm}

Now we will define some further coding schemes, where the sender and
the receiver use correlation as a resource. We will later show that
these coding schemes are very helpful  for the robust and
secure message transmission over an arbitrarily varying wiretap
channel.

\begin{definition}
Let $\mathbf{X}$ and $\mathbf{Y}$ be finite sets.
% Wedenote  the
%sets of joint probability distributions on $\mathbf{X}$ and
%$\mathbf{Y}$ by $P(\mathbf{X}\times\mathbf{Y})$. 
Let  $(X,Y)$  be a
random variable distributed according to a probability distribution
$p\in P(\mathbf{X}\times\mathbf{Y})$.

An $(X,Y)$-correlation-assisted
  $(n, J_n)$  code $\mathcal{C}(X,Y)$ for the
arbitrarily varying classical-quantum wiretap channel $\{(W_t,{V}_t): t \in \theta\}$
consists of a set of stochastic encoders $\left\{E_{\mathbf{x}^n}:
\{ 1,\ldots ,J_n\} \rightarrow P({\mathbf{A}}^n):
\mathbf{x}^n\in\mathbf{X}^n\right\}$, and
 a set of collections of positive semi-definite operators $\Bigl\{\{D_j^{(\mathbf{y}^n)}:
  j = 1,\ldots ,J_n\}:\mathbf{y}^n\in\mathbf{Y}^n\Bigr\}
$ on ${H}^{\otimes n}$ which fulfills $\sum_{j=1}^{J_n} D_j^{(\mathbf{y}^n)} =
\mathrm{id}_{{H}^{\otimes
n}}$ for every $\mathbf{y}^n\in\mathbf{Y}^n$.\end{definition}

\begin{definition}
Let $\mathbf{X}$ and $\mathbf{Y}$ be finite sets,
 and let  $(X,Y)$  be a
random variable distributed according to a joint probability distribution
$p\in P(\mathbf{X}\times\mathbf{Y})$.

 A nonnegative number $R$ is an achievable $m-a-(X,Y)$
   secrecy rate
 (message transmission under the average error criterion using
 $(X,Y)$-correlation-assisted
  $(n, J_n)$  codes) for the arbitrarily
varying  classical-quantum wiretap channel $\{(W_t,{V}_t): t \in \theta\}$ if for every
$\epsilon>0$, $\delta>0$, $\zeta>0$ and sufficiently large $n$ there
exists an $(X,Y)$-correlation-assisted   $(n, J_n)$  code
$\mathcal{C}(X,Y) =
\biggl\{\Bigl(E_{\mathbf{x}^n},\{D_j^{(\mathbf{y}^n)}: j\in\{
1,\ldots ,J_n\}\}\Bigr):  \mathbf{x}^n\in\mathbf{X}^n,\text{ }
\mathbf{y}^n\in\mathbf{Y}^n\biggr\}$  such that $\frac{\log J_n}{n}
> R-\delta$, and
\[\max_{t^n \in \theta^n} \sum_{\mathbf{x}^n\in\mathbf{X}^n}
 \sum_{\mathbf{y}^n\in\mathbf{Y}^n}p(\mathbf{x}^n,\mathbf{y}^n) P_e(\mathcal{C}(\mathbf{x}^n,\mathbf{y}^n), t^n) < \epsilon\text{ ,}\]
\[\max_{t^n\in\theta^n} \chi\left(R_{\mathrm{uni}};Z_{t^n,\mathbf{x}^n}\mid X\right) < \zeta\text{ ,}\]
where $P_e(\mathcal{C}(\mathbf{x}^n,\mathbf{y}^n), t^n)$  is defined as 
\[ P_e(\mathcal{C}(\mathbf{x}^n,\mathbf{y}^n), t^n) := 1- \frac{1}{J_n} \sum_{j=1}^{J_n} 
\mathrm{tr}(W_{t^n}(E_{\mathbf{x}^n}(~|j))D_j^{(\mathbf{y}^n)})\text{
,}\] 
\[\chi\left(R_{\mathrm{uni}};Z_{t^n,\mathbf{x}^n}\mid X\right):= \sum_{\mathbf{y}^n\in\mathbf{Y}^n}p(\mathbf{x}^n,\mathbf{y}^n)
\chi\left(R_{\mathrm{uni}};Z_{t^n,\mathbf{x}^n}\right)\text{ ,}\]
and $Z_{t^n,\mathbf{x}^n}=$ $\biggl\{{V}_{t^n}(E_{\mathbf{x}^n}(~|i)):$ $i\in\{1,\ldots,J_n\}\biggr\}$,
 $p(\mathbf{x}^n,\mathbf{y}^n)=\prod_{i=1}^n
p(\mathbf{x}_i,\mathbf{y}_i)$. Here we allowed $Z_{t^n,\mathbf{x}^n}$, the resulting quantum state
of the wiretapper, to be dependent on $\mathbf{x}^n$,
this means that we do not require $(X,Y)$ to be
secure against eavesdropping.
 \end{definition}

\begin{remark}
Her we follow \cite{Bo/No} and use the definition   ``$m-a-(X,Y)$
   secrecy rate'' because it is important to point out that  here the 
	average error criterion is used.
Please see \cite{Bo/No} for
more discussions on the value of message transmission under the average error criterion
and message transmission under the maximum error criterion.
\end{remark}\vspace{0.2cm}

 \begin{definition}Let  $\Bigl\{C^{\gamma}=\{(E^{\gamma},D_j^{\gamma}):j=1,\ldots,J_n\}:\gamma\in\Lambda\Bigr\}$
be the the set of $(n, J_n)$ deterministic  codes, labeled by a set $\Lambda$.

An  $(n, J_n)$  randomness-assisted quantum code for the
arbitrarily varying classical-quantum wiretap channel
$\{(W_t,{V}_t): t \in \theta\}$ is a distribution $G$ on
$\left(\Lambda,\sigma\right)$, where $\sigma$ is a sigma-algebra so
chosen such that the functions $\gamma \rightarrow P_e(\mathcal{C}^{\gamma},
t^n)$ and  $\gamma \rightarrow \chi
\left(R_{\mathrm{uni}};Z_{\mathcal{C}^{\gamma},t^n}\right)$ are both $G$-measurable
with respect to $\sigma$ for every $t^n\in\theta^n$, where for
$t^n\in\theta^n$ and $\mathcal{C}^{\gamma}=\{(w(j)^{n,\gamma},D_j^{\gamma}):
 j=1,\ldots,J_n\}$,
\[Z_{\mathcal{C}^{\gamma},t^n} :=\left\{{V}_{t^{n}}(w(1)^{n,\gamma}),
{V}_{t^{n}}(w(2)^{n,\gamma}),\ldots,
{V}_{t^{n}}(w(n)^{n,\gamma})\right\}\text{ .}\]
\label{randef}\end{definition}
\begin{remark}
 The randomness-assisted code technique is not to be confused with the random encoding  technique.
 For the  random encoding  technique, only the sender, but not the receiver, 
 randomly chooses a code word in ${\mathbf{A}}^n$ to encode a message $j$ according to a probability distribution.
 The receiver should be able to  decode $j$ even when he only knows the  probability distribution,
 but not which code word is actually chosen by the sender.
 For the randomness-assisted code technique,
 the sender randomly chooses a stochastic encoder $E^{\gamma}$ and the receiver
  chooses a set of the decoder operators $\{D_j^{\gamma'}:j=1,\ldots,J_n\}$.
 The receiver can decode the message if and only  if $\gamma=\gamma'$, i.e.,
 when he knows the sender's randomization.
\end{remark}

\begin{definition}
 Let $\Lambda$ and $\mathcal{C}^{\gamma}$, $\gamma\in \Lambda$, be defined as in Definition \ref{randef}.
 An  $(n, J_n)$  common randomness-assisted
quantum code for the arbitrarily varying classical-quantum wiretap
channel $\{(W_t,{V}_t): t \in \theta\}$ is is a 
 finite subset $\Bigl\{C^{\gamma}=\{(E^{\gamma},D_j^{\gamma}):j=1,\ldots,J_n\}:\gamma\in\Gamma\Bigr\}$
of the set of $(n, J_n)$ deterministic  codes, labeled by a  finite set $\Gamma$.

\end{definition}

\begin{definition}A nonnegative number $R$ is an achievable secrecy rate for the
arbitrarily varying classical-quantum wiretap channel
$\{(W_t,{V}_t): t \in \theta\}$  under randomness-assisted  coding if  for
every
 $\delta>0$, $\zeta>0$, and  $\epsilon>0$, if $n$ is sufficiently large,
there is an $(n, J_n)$  randomness-assisted
quantum code $(\{\mathcal{C}^{\gamma}:\gamma\in \Lambda\},G)$  such that
$\frac{\log J_n}{n} > R-\delta$, and
\[ \max_{t^n\in\theta^n}\int_{\Lambda}P_{e}(\mathcal{C}^{\gamma},t^n)d
G(\gamma) < \epsilon\text{ ,}\]
\[\max_{t^n\in\theta^n} \int_{\Lambda}
\chi\left(R_{\mathrm{uni}},Z_{\mathcal{C}^{\gamma},t^n}\right)dG(\gamma) < \zeta\text{
.}\] Here we allow $Z_{\mathcal{C}^{\gamma},t^n}$, the wiretapper's resulting quantum state, to be dependent on $\mathcal{C}^{\gamma}$. 
This means that we do not require  randomness  to be secure
against eavesdropping.\end{definition}

\begin{definition}A non-negative number $R$ is an achievable secrecy rate for the
arbitrarily varying classical-quantum wiretap channel
$\{(W_t,{V}_t): t \in \theta\}$  under common randomness-assisted  quantum coding if  for
every
 $\delta>0$, $\zeta>0$, and  $\epsilon>0$, if $n$ is sufficiently large,
there is an $(n, J_n)$  common randomness-assisted
quantum code $(\{\mathcal{C}^{\gamma}:\gamma\in \Gamma\})$  such that
$\frac{\log J_n}{n} > R-\delta$, and
\[ \max_{t^n\in\theta^n} \frac{1}{\left|\Gamma\right|} \sum_{\gamma=1}^{\left|\Gamma\right|}P_{e}(\mathcal{C}^{\gamma},t^n) < \epsilon\text{ ,}\]
\[\max_{t^n\in\theta^n} \chi\left(R_{\mathrm{uni}},Z_{\mathcal{C}^{\gamma},t^n}\mid \Gamma\right) < \zeta\text{ ,}\] 
where
\[\chi\left(R_{\mathrm{uni}},Z_{\mathcal{C}^{\gamma},t^n}\mid \Gamma\right):=\frac{1}{\left|\Gamma\right|} \sum_{\gamma=1}^{\left|\Gamma\right|}
\chi\left(R_{\mathrm{uni}},Z_{\mathcal{C}^{\gamma},t^n}\right)\text{ .}\]

This
means that we do not require the common randomness  to be secure
against eavesdropping.\end{definition}

We may consider the deterministic code, the
$(X,Y)$-correlation-assisted   code, the $((X,Y),r)$-correlation-assisted
 code, the $(X,Y)$-correlation-assisted  $(n, J_n)$
code, and   the common randomness-assisted quantum code
 as  special cases of the  randomness-assisted quantum code.
This means that  randomness is a stronger resource than both
  common randomness and  the $(X,Y)$-correlation, in the sense that
it requires more randomness than   common randomness and  the $(X,Y)$-correlation.
Randomness is therefore a more ``costly'' resource.

\begin{definition}
Let $\{(W_t,{V}_t): t \in \theta\}$ be an arbitrarily
varying classical-quantum wiretap channel.\\[0.2cm]
The  supremum of  all  achievable  (deterministic) secrecy rates of
$\{(W_t,{V}_t): t \in \theta\}$ is called the (deterministic)  secrecy
capacity of $\{(W_t,{V}_t): t \in \theta\}$, denoted by
$C_s(\{(W_t,{V}_t): t \in \theta\})$.\\[0.15cm]
 The  supremum of  all  achievable $m-a-(X,Y)$
   secrecy rates  of
$\{(W_t,{V}_t): t \in \theta\}$  is called the
 $m-a-(X,Y)$
      secrecy capacity, denoted by
$C_s(\{(W_t,{V}_t): t \in \theta\};corr(X,Y))$.\\[0.15cm]
  The
 supremum of  all    achievable secrecy rates   under  random-assisted  quantum
coding  of $\{(W_t,{V}_t): t \in \theta\}$ is called the
 random-assisted   secrecy capacity of
$\{(W_t,{V}_t): t \in \theta\}$, denoted by
$C_s(\{(W_t,{V}_t): t \in \theta\};r)$.\\[0.15cm]
  The  supremum of  all  achievable secrecy rates
under common randomness-assisted  quantum coding  of
$\{(W_t,{V}_t): t \in \theta\}$ is called the common randomness-assisted
   secrecy capacity of $\{(W_t,{V}_t): t \in \theta\}$,
denoted by $C_s(\{(W_t,{V}_t): t \in \theta\};cr)$.\end{definition}
% Let $\mathbf{r}\in\mathbb{N}$, the tight upper bound on achievable $m-a-((X,Y),\mathbf{r})$
 %secrecy rate  of
%$\{(W_t,{V}_t): t \in \theta\}$  is called the
% $m-a-(X,Y)$
 %     secrecy capacity of
%$\{(W_t,{V}_t): t \in \theta\}$ , denoted by
%$C_s(\{(W_t,{V}_t): t \in \theta\};corr_{\mathbf{r}}(X,Y))$.\\[0.15cm]
% The tight upper bound on achievable $m-a-((X,Y))$
%causal secrecy rate  of $\{(W_t,{V}_t): t \in \theta\}$  is called the
% $m-a-(X,Y)$
%causal    secrecy capacity of $\{(W_t,{V}_t): t \in \theta\}$, denoted
%by
%$C_s(\{(W_t,{V}_t): t \in \theta\};corr_{causal}(X,Y))$.\\[0.15cm]

For an arbitrarily varying classical-quantum wiretap channel
$\{(W_t,{V}_t): t \in \theta\}$ and random variable  $(X,Y)$
distributed on finite sets $\mathbf{X}$ and $\mathbf{Y}$, the
following facts are obvious and follow from the definitions.

%\begin{align}&C_s(\{(W_t,{V}_t): t \in \theta\})\notag\\
%&\leq C_s((W_t,{V}_t)_{t \in
%\theta};corr_{causal}(X,Y))\notag\\
%&\leq C_s((W_t,{V}_t)_{t \in
%\theta};corr(X,Y))\notag\\
%&\leq C_s(\{(W_t,{V}_t): t \in \theta\};r)\text{ ,}\end{align}

\begin{align}&C_s(\{(W_t,{V}_t): t \in \theta\})\notag\\
&\leq C_s((W_t,{V}_t)_{t \in
\theta};corr(X,Y))\notag\\
&\leq C_s(\{(W_t,{V}_t): t \in \theta\};r)\text{ ,}\end{align}

%\begin{align}&C_s(\{(W_t,{V}_t): t \in \theta\})\notag\\
%&\leq C_s((W_t,{V}_t)_{t \in
%\theta};corr_{\mathbf{r}}(X,Y))\notag\\
%&\leq C_s((W_t,{V}_t)_{t \in
%\theta};corr(X,Y))\notag\\
%&\leq C_s(\{(W_t,{V}_t): t \in \theta\};r)\text{ ,}\end{align}
\begin{equation}C_s(\{(W_t,{V}_t): t \in \theta\})\leq C_s((W_t,{V}_t)_{t \in
\theta};cr)\leq C_s(\{(W_t,{V}_t): t \in \theta\};r)\text{ .}\end{equation}\vspace{0.2cm}

\section{Ahlswede dichotomy for Arbitrarily Varying Classical-Quantum Wiretap Channels}\label{ADFAVCQWC}

In this section, we analyze the secrecy capacities of various coding schemes with resource
assistance. Our goal is to see what the
effects are on the secrecy capacities of an arbitrarily
varying classical-quantum wiretap channel if we use deterministic code, randomness-assisted code,  or
 common randomness-assisted code.

\begin{theorem}[Ahlswede dichotomy] Let $\{(W_t,{V}_t): t \in \theta\}$
 be an arbitrarily
varying classical-quantum wiretap channel.
\begin{enumerate}\item
\begin{enumerate}
\item If the
arbitrarily varying  classical-quantum channel $\{W_t : t \in \theta\}$
is not symmetrizable, then \begin{equation}
C_s(\{(W_t,{V}_t): t \in \theta\})=C_s(\{(W_t,{V}_t) : t \in \theta\};r)\text{ .}
\end{equation}\label{dichpart1a}
\item If $\{W_t : t \in \theta\}$ is symmetrizable,
 \begin{equation}
C_s(\{(W_t,{V}_t): t \in \theta\})=0\text{ .}
\end{equation} \label{dichpart1b}\end{enumerate}
\label{dichpart1}
\item
\begin{equation}
C_s(\{(W_t,{V}_t): t \in \theta\};cr)=C_s(\{(W_t,{V}_t): t \in
\theta\};r)\text{ .}
\end{equation}\label{dichpart2}\end{enumerate}
\label{dichpart}
\end{theorem}

\begin{proof}

Our proof is similar to the proof of Ahlswede dichotomy for  arbitrarily varying classical-quantum channels in
\cite{Ahl/Bli}. The different between our proof and the proofs in \cite{Ahl/Bli} is that we have to
additionally consider the security.

\subsection{Proof of Theorem \ref{dichpart}. \ref{dichpart2}}

 At first we use
random encoding technique to show the existence of a common randomness-assisted code.\vspace{0.2cm}

Choose arbitrary positive $\epsilon$ and $\zeta$.
Assume we have an   $(n, J_n)$  randomness-assisted
 code $(\{\mathcal{C}^{\gamma}:\gamma\in \Lambda\},G)$ for
$\{(W_t,{V}_t): t \in \theta\}$ such that
\[ \max_{t^n\in\theta^n}\int_{\Lambda}P_{e}(\mathcal{C}^{\gamma},t^n)d
G(\gamma) < \epsilon\text{ ,}\]
\[\max_{t^n\in\theta^n} \int_{\Lambda}
\chi\left(R_{\mathrm{uni}},Z_{\mathcal{C}^{\gamma},t^n}\right)dG(\gamma) < \zeta\text{
.}\]\vspace{0.15cm}

 Consider
now $n^3$ independent and identically distributed random variables
$\overline{\mathcal{C}}_1,\overline{\mathcal{C}}_2,\ldots,\overline{\mathcal{C}}_{n^3}$ with values in $\{\mathcal{C}^{\gamma}:\gamma\in
\Lambda\}$ such that $Pr(\overline{\mathcal{C}}_i=\mathcal{C})=G(\mathcal{C})$ for all $ \mathcal{C}
\in\{\mathcal{C}^{\gamma}:\gamma\in \Lambda\}$ and for all
$i\in\{1,\ldots,n^3\}$. For a fixed $t^n \in \theta^n$ we have
\begin{align}&P\left(\sum_{i=1}^{n^3}\chi\left(R_{\mathrm{uni}},Z_{\overline{\mathcal{C}}_i,t^n}\right)>n^3\lambda \right)\notag\\
&=P\left(\exp\left(\sum_{i=1}^{n^3} \frac{1}{n} 2 \chi\left(R_{\mathrm{uni}},Z_{\overline{\mathcal{C}}_i,t^n}\right)\right)>\exp( \frac{1}{n}2n^3\lambda )\right)\notag\\
&\leq\exp\left( -2n^2\lambda\right)\prod_{i=1}^{n^3}\mathbb{E}_G\exp\left(\frac{1}{n} 2\chi\left(R_{\mathrm{uni}},Z_{\overline{\mathcal{C}}_i,t^n}\right)\right)\notag\\
&=\exp\left( -2n^2\lambda\right)\mathbb{E}_G\exp\left(\sum_{i=1}^{n^3}\frac{1}{n} 2 \chi\left(R_{\mathrm{uni}},Z_{\overline{\mathcal{C}}_i,t^n}\right)\right)\notag\\
&\leq\exp\left( -2 n^2\lambda\right)\prod_{i=1}^{n^3}\mathbb{E}_G\left[1+\sum_{k=1}^{\infty} \frac{2^k \frac{1}{n}\chi\left(R_{\mathrm{uni}},Z_{\overline{\mathcal{C}}_i,t^n}\right)}{k!}\right]\notag\\
&=\exp\left( -2 n^2\lambda\right)\left[1+\sum_{k=1}^{\infty}  \frac{2^k \frac{1}{n}\mathbb{E}_G\chi\left(R_{\mathrm{uni}},Z_{\overline{\mathcal{C}}_i,t^n}\right)}{k!}\right]^{n^3}\notag\\
&\leq\exp\left( -2 n^2\lambda\right)\left[1+\sum_{k=1}^{\infty}  \frac{2^k\epsilon}{nk!}\right]^{n^3}\notag\\
&=\exp\left( -2 n^2\lambda\right)\left[1+\frac{1}{n}\epsilon\exp 2
\right]^{n^3}\text{ ,}\label{align1f}\end{align} the second
inequality holds because the right side  is part of the Taylor
series.\vspace{0.15cm}

We fix $n\in\mathbb{N}$ and define \[h_n(x):=n\log(1+\frac{1}{n}e^2x)-x\text{ .}\]
We have $h_n(0)=0$ and\begin{align*}&h_n'(x)\\
  &=n\frac{1}{1+\frac{1}{n}e^2x}\frac{1}{n}e^2-1\\
  &=\frac{ne^2}{e^2x+n}-1\text{ .}
 \end{align*}

$\frac{ne^2}{e^2x+n}-1$ is positive if $x< \frac{e^2-1}{e}n$, thus
if $\hat{c}< \frac{e^2-1}{e}n$, $h_n(x)$ is  strictly monotonically increasing
in the interval $]0,\hat{c}]$. Thus  $h_n(x)$ is positive for $0<x\leq \hat{c}$.
For every positive $\hat{c}$, $\hat{c}< \frac{e^2-1}{e}n$ holds if
$n> \frac{e}{e^2-1}\hat{c}$.  Thus for any positive $\epsilon$, $\epsilon\leq
n\log(1+\frac{1}{n}\epsilon\exp 2)$ if $n$
is large enough.
Choose $\lambda\leq\epsilon$  and let $n$ be sufficiently large,
we have $\lambda\leq n\log(1+\frac{1}{n}\epsilon\exp 2)$, therefore
\begin{align}&\exp\left( -2 \lambda n^2\right)\left[1+\frac{1}{n}\epsilon\exp 2 \right]^{n^3}\notag\\
 &=\exp\left( - \lambda n^2\right)\exp\left(n^2( - \lambda + n\log(1+\frac{1}{n}\epsilon\exp 2 ))\right)\notag\\
 &\leq\exp\left( - \lambda n^2\right)\text{ .}\label{litalign}
\end{align}

By (\ref{align1f}) and  (\ref{litalign})
\begin{align}&P\left(\sum_{i=1}^{n^3} \chi\left(R_{\mathrm{uni}},Z_{\overline{\mathcal{C}}_i,t^n}\right)>\lambda n^3\text{ }\forall t^n \in \theta^n\right)\notag\\
&<|\theta|^n\exp(-\lambda n^2)\notag\\
&=\exp(n\log |\theta|-\lambda n^2)\notag\\
&=\exp(-n\lambda )  \text{ .}\end{align} \vspace{0.15cm}

When $P_{er} < \zeta$ holds, in a similar way as $(\ref{align1f})$,
choose $\lambda\leq\zeta$, we can show that
\begin{equation}P\left(\sum_{i=1}^{n^3} P_e(\overline{\mathcal{C}}_i,t^n)>\lambda n^3\text{ }\forall t^n\in\theta^n\right)<e^{-\lambda n}
\text{ .}\label{ln2f}\end{equation}   \vspace{0.15cm}

Let $\lambda:=\min\{\epsilon, \zeta\}$, we have \begin{align*}&P\left(\sum_{i=1}^{n^3}
P_e(\overline{\mathcal{C}}_i,t^n)>\lambda n^3 \text{ or } \sum_{i=1}^{n^3}
 \chi\left(R_{\mathrm{uni}},Z_{\overline{\mathcal{C}}_i,t^n}\right)>\lambda n^3\text{ }\forall t^n\in\theta^n\right) \leq    2 e^{-\lambda n^3}  \text{ .} \end{align*}

We denote the event \begin{align*}&\mathbf{E}_n :=\biggl\{\mathcal{C}_1,\mathcal{C}_2,\ldots,\mathcal{C}_{n^3}\in
\mathcal{C}_{\nu}' :
\frac{1}{n^3}\sum_{i=1}^{n^3} P_e(\mathcal{C}_i,t^n)\leq\lambda\\
 &\text{ and }
\frac{1}{n^3}\sum_{i=1}^{n^3}
\chi\left(R_{\mathrm{uni}},Z_{\mathcal{C}_i,t^n}\right)\leq\lambda\biggr\}\text{ .} \end{align*}
If $n$ is large enough, then $P(\mathbf{E}_n)$ is positive. This
means $\mathbf{E}_n$ is not the
empty set, since $P(\emptyset)=0$ by  definition.
Thus there exist  codes 
$\mathcal{C}_i=
\left(E_{i}^{n},\left\{D_{j,i}^{n}: j =1, \ldots,
J_n\right\}\right)$ $\in
\mathcal{C}_{\nu}'$ for $i\in\{1,\ldots,n^3\}$ with a positive probability
such that
\begin{equation} \frac{1}{n^3}\sum_{i=1}^{n^3} P_e(\mathcal{C}_i,t^n)<\lambda \text{ and }
\frac{1}{n^3}\sum_{i=1}^{n^3}
\chi\left(R_{\mathrm{uni}},Z_{\mathcal{C}_i,t^n}\right)\leq\lambda\text{
.}\label{n2'f}\end{equation}

By (\ref{n2'f}), for any $n\in \mathbb{N}$ and positive $\lambda$, if there is an   $(n, J_n)$  randomness-assisted
 code $(\{\mathcal{C}^{\gamma}:\gamma\in \Lambda\},G)$ for
$\{(W_t,{V}_t): t \in \theta\}$ such that
\[ \max_{t^n\in\theta^n}\int_{\Lambda}P_{e}(\mathcal{C}^{\gamma},t^n)d
G(\gamma) < \lambda\text{ ,}\]
\[\max_{t^n\in\theta^n} \int_{\Lambda}
\chi\left(R_{\mathrm{uni}},Z_{\mathcal{C}^{\gamma},t^n}\right)dG(\gamma) < \lambda\text{
,}\]
there is also an   $(n, J_n)$ common randomness-assisted
 code $\left\{\mathcal{C}_1,\mathcal{C}_2,\ldots,\mathcal{C}_{n^3}\right\}$ such that
\[ \max_{t^n\in\theta^n}\frac{1}{n^3}\sum_{i=1}^{n^3}P_{e}(\mathcal{C}_{i},t^n)< \lambda\text{ ,}\]
\[\max_{t^n\in\theta^n} \frac{1}{n^3}\sum_{i=1}^{n^3}
\chi\left(R_{\mathrm{uni}},Z_{\mathcal{C}_{i},t^n}\right) < \lambda\text{
.}\]
Therefore we have
\[C_s(\{(W_t,{V}_t): t \in \theta\};cr)\geq C_s(\{(W_t,{V}_t): t \in \theta\};r) \text{ .}\]
This and the fact that
\[C_s(\{(W_t,{V}_t): t \in \theta\};cr)\leq C_s(\{(W_t,{V}_t): t \in \theta\};r) \text{ ,}\]
prove  Theorem \ref{dichpart}. \ref{dichpart2}.   \vspace{0.3cm}

\subsection{Proof of  Theorem \ref{dichpart}.  \ref{dichpart1a}}

Now we are going to use Theorem \ref{dichpart}. \ref{dichpart2} to
 prove  Theorem \ref{dichpart}. \ref{dichpart1a}.

To show the lower bound in 
 Theorem \ref{dichpart}. \ref{dichpart1a},  we build a two-part code word, which consists of a non-secure code word
 and a   common randomness-assisted secure code word. The non-secure  one is used to
create the common randomness for the sender and the legal receiver.
The  common randomness-assisted secure code word is used to transmit the message to the legal
receiver.
\vspace{0.2cm}

Choose arbitrary positive $\epsilon$ and $\zeta$.
Assume we have an   $(n, J_n)$  randomness-assisted
 code $(\{\mathcal{C}^{\gamma}:\gamma\in \Lambda\},G)$ for
$\{(W_t,{V}_t): t \in \theta\}$ such that
\[ \max_{t^n\in\theta^n}\int_{\Lambda}P_{e}(\mathcal{C}^{\gamma},t^n)d
G(\gamma) < \epsilon\text{ ,}\]
\[\max_{t^n\in\theta^n} \int_{\Lambda}
\chi\left(R_{\mathrm{uni}},Z_{\mathcal{C}^{\gamma},t^n}\right)dG(\gamma) < \zeta\text{
,}\] by  Theorem \ref{dichpart}. \ref{dichpart2},
there is also an   $(n, J_n)$ common randomness-assisted
 code $\left\{\mathcal{C}_1,\mathcal{C}_2,\ldots,\mathcal{C}_{n^3}\right\}$ such that
\begin{equation} \max_{t^n\in\theta^n}\frac{1}{n^3}\sum_{i=1}^{n^3}P_{e}(\mathcal{C}_{i},t^n)<
\lambda\text{ ,}\label{anf1}\end{equation}
\begin{equation}\max_{t^n\in\theta^n} \frac{1}{n^3}\sum_{i=1}^{n^3}
\chi\left(R_{\mathrm{uni}},Z_{\mathcal{C}_{i},t^n}\right) < \lambda\text{
,}\label{anf2}\end{equation} where $\lambda:=\min\{\epsilon, \zeta\}$.\vspace{0.2cm}

If the
arbitrarily varying  classical-quantum channel $\{W_{t}:
x\in\mathcal{X}\}$ is not symmetrizable, then by \cite{Ahl/Bli},
the capacity for message transmission of $\{W_{t}:
x\in\mathcal{X}\}$ is positive. By Remark \ref{detvsran}
we may assume that the capacity for message transmission of $\{W_{t}:
x\in\mathcal{X}\}$ using deterministic encoder is positive.
This means
for any positive $\vartheta$,
if $n$ is sufficiently large, there is a code $\biggl(\Bigl(c^{\mu(n)}_i\Bigr)_{i\in\{1,\ldots,n^3\}},\{D_i^{\mu(n)}:
i\in\{1,\ldots,n^3\}\}\biggr)$ with deterministic encoder of length $\mu(n)$, where $2^{\mu(n)}=o(n)$
such that\begin{equation}1- \frac{1}{n^3} \sum_{i=1}^{n^3}
\mathrm{tr}(W_{t^n}(c^{\mu(n)}_i)D_i^{\mu(n)})\leq \vartheta\text{ .}\label{anf3}\end{equation}

%Without loss of generality, we may assume that this code is not
%secure against the wiretapper, i.e. the wiretapper has the full
%knowledge of $i$, because in the other case the sender and the
%receiver can use $\biggl(\Bigl(c^{\mu(n)}_i\Bigr)_{i\in\{1,\ldots,n^3\}},\{D_i^{\mu(n)}:
%i\in\{1,\ldots,n^3\}\}\biggr)$ to create a secret
%key, c.f. \cite{Ahl/Csi1} and \cite{Ahl/Csi2}, to get even a larger
%capacity than the deterministic secrecy capacity without  a secret
%key.

%Since this code is not secure against the wiretapper,
%we can assume that the wiretapper can find a decoding
%scheme such that for any
%positive $\delta$, if $n$ is large enough,
%\[\min_{t^{\mu(n)}\in\theta^{\mu(n)}} I(Y_{\mathrm{uni}},Z_{t^{\mu(n)}}) \geq  H(Y_{\mathrm{uni}}) -\delta\text{
%,}\]where $Y_{\mathrm{uni}}$ is a
%uniformly distributed random variable with value in $\{1,\ldots
%,n^3\}$, and $Z_{t^{\mu(n)}}$ is the wiretapper's resulting random output
%after decoding.
%Since
%the Holevo bound is an upper bound for the mutual
%information of all decoding schemes, we have
%\begin{equation}\min_{t^{\mu(n)}\in\theta^{\mu(n)}}\chi(Y_{\mathrm{uni}},
%\mathcal{P}_{t^{\mu(n)}})\ \geq  H(Y_{\mathrm{uni}}) -\delta\text{
%,}\label{codet2}
%\end{equation} where $\mathcal{P}_{t^{\mu(n)}}:=
%\left\{V_{t^{\mu(n)}}({c}^{\mu(n)}_1),V_{t^{\mu(n)}}({c}^{\mu(n)}_2),\ldots,
%V_{t^{\mu(n)}}({c}^{\mu(n)}_{n^3})\right\}$ is the wiretapper's resulting
%quantum state at channel state $t^{\mu(n)}$.\vspace{0.2cm}

We now can construct a code
$\mathcal{C}^{\mathrm{det}} $ $=$ $\biggl(E^{\mu(n)+n},\Bigl\{D_{j}^{\mu(n)+n} :
j=1,\ldots,J_n\Bigr\}\biggr)$, where for $a^{\mu(n)
+n} = (a^{\mu(n)},a^n)\in{\mathbf{A}}^{\mu(n)+n}$ \[E^{\mu(n)+n}(a^{\mu(n)
+n}|j)=\begin{cases}
  \frac{1}{n^3}E^{n}_{i}(a^{n}|j) \text{ if } a^{\mu(n)} = c^{\mu(n)}_i\\
    0 \text{ else}             \end{cases}\text{ ,}
\] and \[D_{j}^{\mu(n)+n} := \sum_{i=1}^{n^3} D_i^{\mu(n)}\otimes D_{i,j}^{n} \text{ .}
\]
It is a composition of  the  code $\Bigl(c^{\mu(n)}_i\Bigr)_{i=1,\ldots,n^3},
\{D_i^{\mu(n)}: i=1,\ldots,n^3\})$ and the  code
$\mathcal{C}_i=(E_{i}^{n}, \{D_{i,j}^{n}:j=1,\ldots,J_n\}$. This is a  code
of length $\mu(n) +n$.

\subsubsection{This code is secure against eavesdropping}

We are going to show that the two-part code word is secure when the
common randomness-assisted part is secure. Since the two-part code
can be seen as a function of its common randomness-assisted part
the idea is similar to applying the quantum 
data processing inequality (cf. \cite{Wil}) when
we consider quantum mutual information as security criterion. \vspace{0.2cm}

For any $i\in \{1,\ldots ,n^3\}$ let
\[\mathfrak{Z}_{i,t^{\mu(n)+n}} :=\biggl\{{V}_{t^{\mu(n)}}(c^{\mu(n)}_i)\otimes{V}_{t^{n}}
(  E_{i}^{n}(~\mid 1)),\ldots, {V}_{t^{\mu(n)}}(c^{\mu(n)}_i)\otimes{V}_{t^{n}}(E_{i}^{n}(~\mid J_n))\biggr\}\text{ .}\]
For any
$t^{\mu(n) +n}= (t^{\mu(n) },t^{n})$  we have
\begin{align}& \chi\left(R_{\mathrm{uni}},\mathfrak{Z}_{i,t^{\mu(n)+n}}\right)\notag\\
&= S\left(\frac{1}{J_n}
 \sum_{j=1}^{J_n}\sum_{a^n\in {\mathbf{A}}^n} E_{i}^{n}(a^n\mid j){V}_{t^{\mu(n)}}(c^{\mu(n)}_i)\otimes{V}_{t^{n}}( a^n)\right)\notag\\
&-\frac{1}{J_n}
 \sum_{j=1}^{J_n}S\left(\sum_{a^n\in {\mathbf{A}}^n} E_{i}^{n}(a^n\mid j){V}_{t^{\mu(n)}}(c^{\mu(n)}_i)\otimes{V}_{t^{n}}( a^n)\right)\notag\\
&= S\left( {V}_{t^{\mu(n)}}(c^{\mu(n)}_i)\right)+ S\left(\frac{1}{J_n}
\sum_{j=1}^{J_n}\sum_{a^n\in {\mathbf{A}}^n} E_{i}^{n}(a^n\mid j){V}_{t^n}(a^n)\right)-S\left( {V}_{t^{\mu(n)}}(c^{\mu(n)}_i)\right)\notag\\
&- \frac{1}{J_n}
\sum_{j=1}^{J_n}S\left(\sum_{a^n\in {\mathbf{A}}^n} E_{i}^{n}(a^n\mid j){V}_{t^n}(a^n)\right)\notag\\
&=S\left(\frac{1}{J_n}
\sum_{j=1}^{J_n} {V}_{t^n}(E_{i}^{n}(~\mid j))\right)-\frac{1}{J_n}
\sum_{j=1}^{J_n} S\left({V}_{t^n}(E_{i}^{n}(~\mid j))\right)\notag\\
 &= \chi\left(R_{\mathrm{uni}},Z_{\mathcal{C}_i,t^n}\right)\text{ .}\label{einfach}\end{align}

By definition,  we have\begin{align*}&Z_{\mathcal{C}^{\mathrm{det}},t^{\mu(n) +n}} \\
&= \biggl\{{V}_{t^{\mu(n) +n}}( E^{\mu(n)+n}(~\mid 1)),\ldots,
{V}_{t^{\mu(n) +n}}( E^{\mu(n)+n}(~\mid J_n))\biggr\}\\
&=\biggl\{\frac{1}{n^3}
\sum_{i=1}^{n^3}\sum_{a^n\in {\mathbf{A}}^n} E_{i}^{n}(a^n\mid 1){V}_{t^{\mu(n) +n}}\Bigl(\left(c^{\mu(n)}_i,
a^n\right)\Bigr),\ldots,\\
& \frac{1}{n^3}
\sum_{i=1}^{n^3}\sum_{a^n\in {\mathbf{A}}^n} E_{i}^{n}(a^n\mid J_n){V}_{t^{\mu(n) +n}}\Bigl(\left(c^{\mu(n)}_i,
a^n\right)\Bigr)\biggr\}\\
&= \biggl\{  \frac{1}{n^3}\sum_{i=1}^{n^3}\sum_{a^n\in {\mathbf{A}}^n} E_{i}^{n}(a^n\mid 1){V}_{t^{\mu(n) }}(c^{\mu(n)}_i)\otimes
 {V}_{t^{n}}(a^n)   ,\ldots,\\
& \frac{1}{n^3}\sum_{i=1}^{n^3}\sum_{a^n\in {\mathbf{A}}^n} E_{i}^{n}(a^n\mid J_n){V}_{t^{\mu(n) }}(c^{\mu(n)}_i)\otimes
 {V}_{t^{n}}(a^n) \biggr\}\\
&=\biggl\{  \frac{1}{n^3}\sum_{i=1}^{n^3}{V}_{t^{\mu(n) }}(c^{\mu(n)}_i)\otimes
 {V}_{t^{n}}( E_{i}^{n}(~\mid 1))     ,\ldots,\\
& \frac{1}{n^3}\sum_{i=1}^{n^3} {V}_{t^{\mu(n) }}(c^{\mu(n)}_i)\otimes
 {V}_{t^{n}}(E_{i}^{n}(~\mid J_n))    \biggr\}\text{ .}\end{align*}

% We denote
%\begin{align*}&\mathcal{Q}_{t^{\mu(n) +n}} :=\\
%&\biggl\{\frac{1}{J_n}
% \sum_{j=1}^{J_n}\sum_{a^n\in {\mathbf{A}}^n} E_{1}^{n}(a^n\mid j){V}_{t^{\mu(n) +n}}\Bigl(\left(c^{\mu(n)}_1,
%a^n\right)\Bigr),\ldots,\\
% &\frac{1}{J_n}
% \sum_{j=1}^{J_n}\sum_{a^n\in {\mathbf{A}}^n} E_{n^3}^{n}(a^n\mid j){V}_{t^{\mu(n) +n}}\Bigl(\left(c^{\mu(n)}_{n^3},
%a^n\right)\Bigr)\biggr\}\text{ ,}\end{align*}
%and for every $j\in\{1,\ldots,J_n\}$ we denote
%\begin{align*}&\mathcal{Q}_{j,t^{\mu(n) +n}} :=\\
%&\biggl\{\sum_{a^n\in {\mathbf{A}}^n} E_{1}^{n}(a^n\mid j){V}_{t^{\mu(n) +n}}\Bigl(\left(c^{\mu(n)}_1,
%a^n\right)\Bigr),\ldots,\\
% &\sum_{a^n\in {\mathbf{A}}^n} E_{n^3}^{n}(a^n\mid j){V}_{t^{\mu(n)
%+n}}\Bigl(\left(c^{\mu(n)}_{n^3},
%a^n\right)\Bigr)\biggr\}\text{ .}\end{align*}

By (\ref{anf2}) and (\ref{einfach})  for any
$t^{\mu(n) +n}= (t^{\mu(n) }t^{n})$  we have

\begin{align}& \chi\left(R_{\mathrm{uni}},Z_{\mathcal{C}^{\mathrm{det}},t^{\mu(n) +n}}\right)\notag\\
&\leq
\chi\left(R_{\mathrm{uni}},Z_{\mathcal{C}^{\mathrm{det}},t^{\mu(n) +n}}\right)-\frac{1}{n^3}\sum_{i=1}^{n^3}
\chi\left(R_{\mathrm{uni}},Z_{\mathcal{C}_i,t^n}\right)+\lambda \notag\\
&=
\chi\left(R_{\mathrm{uni}},Z_{\mathcal{C}^{\mathrm{det}},t^{\mu(n) +n}}\right)-\frac{1}{n^3}\sum_{i=1}^{n^3}
\chi\left(R_{\mathrm{uni}},\mathfrak{Z}_{i,t^{\mu(n)+n}}\right)+\lambda \notag\\
 &=S\left(\frac{1}{J_n}\frac{1}{n^3}
 \sum_{j=1}^{J_n}\sum_{i=1}^{n^3}{V}_{t^{\mu(n)}}(c^{\mu(n)}_i)\otimes  {V}_{t^{n}}(E_{i}^{n}(~\mid j))\right)\notag\\
&-\frac{1}{J_n}
 \sum_{j=1}^{J_n}S\left(\frac{1}{n^3}\sum_{i=1}^{n^3}{V}_{t^{\mu(n) }}(c^{\mu(n)}_i)\otimes
 {V}_{t^{n}}(E_{i}^{n}(~\mid j))\right)\notag\\
&-\frac{1}{n^3}\sum_{i=1}^{n^3}S\left(\frac{1}{J_n}
 \sum_{j=1}^{J_n}{V}_{t^{\mu(n) }}(c^{\mu(n)}_i)\otimes
 {V}_{t^{n}}(E_{i}^{n}(~\mid j))\right)\notag\\
&+\frac{1}{J_n}\frac{1}{n^3}
 \sum_{j=1}^{J_n}\sum_{i=1}^{n^3}S\left({V}_{t^{\mu(n)}}(c^{\mu(n)}_i)\otimes
 {V}_{t^{n}}(E_{i}^{n}(~\mid j))\right)+\lambda\text{ .}\label{secudet1}\end{align} \vspace{0.2cm}

 Let $H^{\mathfrak{H}}$ be a $n^3$-dimensional Hilbert space,
spanned by an orthonormal basis $\{|i\rangle : i = 1, \ldots, n^3\}$. 
Let $H^{\mathfrak{J}}$ be a $J_n$ dimensional Hilbert space, spanned by an orthonormal basis 
$\{|j\rangle : j = 1, \ldots, J_n\}$. 
We define
\[\varphi^{\mathfrak{J}\mathfrak{H}H^{\mu(n)+n}}:=\frac{1}{J_n}\frac{1}{n^3}\sum_{j=1}^{J_n}\sum_{i=1}^{n^3}
|j\rangle\langle j|\otimes|i\rangle\langle i|\otimes
{V}_{t^{\mu(n)}}(c^{\mu(n)}_i)\otimes {V}_{t^{ n}}(E_{i}^{n}(~\mid j))\text{ .}\]

We have
\begin{align*}&\varphi^{\mathfrak{J}H^{\mu(n)+n}}=\mathrm{tr}_{\mathfrak{H}}\left(\varphi^{\mathfrak{J}\mathfrak{H}H^{\mu(n)+n}} \right)=\frac{1}{J_n}\frac{1}{n^3}\sum_{j=1}^{J_n}\sum_{i=1}^{n^3}|j\rangle\langle j|\otimes
{V}_{t^{\mu(n)}}(c^{\mu(n)}_i)\otimes {V}_{t^{ n}}(E_{i}^{n}(~\mid j))\text{ ,}\end{align*}
\begin{align*}&\varphi^{\mathfrak{H}H^{\mu(n)+n}}=\mathrm{tr}_{\mathfrak{J}}\left(\varphi^{\mathfrak{J}\mathfrak{H}H^{\mu(n)+n}} \right)=\frac{1}{J_n}\frac{1}{n^3}\sum_{j=1}^{J_n}\sum_{i=1}^{n^3} |i\rangle\langle i|\otimes
{V}_{t^{\mu(n)}}(c^{\mu(n)}_i)\otimes {V}_{t^{ n}}(E_{i}^{n}(~\mid j))\text{ ,}\end{align*}

\begin{align*}&\varphi^{H^{\mu(n)+n}}=\mathrm{tr}_{\mathfrak{J}\mathfrak{H}}\left(\varphi^{\mathfrak{J}\mathfrak{H}H^{\mu(n)+n}} \right)=\frac{1}{J_n}\frac{1}{n^3}\sum_{j=1}^{J_n}\sum_{i=1}^{n^3}
{V}_{t^{\mu(n)}}(c^{\mu(n)}_i)\otimes {V}_{t^{ n}}(E_{i}^{n}(~\mid j))\text{ .}\end{align*}

Furthermore
\begin{align*}&S(\varphi^{\mathfrak{J}H^{\mu(n)+n}})\\
&=S\left(\frac{1}{J_n}\sum_{j=1}^{J_n}\sum_{i=1}^{n^3}|j\rangle\langle j|\otimes
{V}_{t^{\mu(n)}}(c^{\mu(n)}_i)\otimes {V}_{t^{ n}}(E_{i}^{n}(~\mid j))\right)\\
& = H(R_{\mathrm{uni}})+ \frac{1}{J_n}\sum_{j=1}^{J_n}S\left(\frac{1}{n^3}\sum_{i=1}^{n^3}
{V}_{t^{\mu(n)}}(c^{\mu(n)}_i)\otimes {V}_{t^{ n}}(E_{i}^{n}(~\mid j))\right)\text{ ,}\end{align*}

\begin{align*}&S(\varphi^{\mathfrak{H}H^{\mu(n)+n}})\\
&=S\left(\frac{1}{J_n}\frac{1}{n^3}\sum_{j=1}^{J_n}\sum_{i=1}^{n^3} |i\rangle\langle i|\otimes
{V}_{t^{\mu(n)}}(c^{\mu(n)}_i)\otimes {V}_{t^{ n}}(E_{i}^{n}(~\mid j))\right)\\
& =H(Y_{\mathrm{uni}})+ \frac{1}{n^3}\sum_{i=1}^{n^3}S\left(\frac{1}{J_n}\sum_{j=1}^{J_n}
{V}_{t^{\mu(n)}}(c^{\mu(n)}_i)\otimes {V}_{t^{ n}}(E_{i}^{n}(~\mid j))\right)\text{ ,}\end{align*}

\begin{align*}&S(\varphi^{\mathfrak{J}\mathfrak{H}H^{\mu(n)+n}})\\
&=S\left(\frac{1}{J_n}\frac{1}{n^3}\sum_{j=1}^{J_n}\sum_{i=1}^{n^3}|j\rangle\langle j|\otimes|i\rangle\langle i|\otimes
{V}_{t^{\mu(n)}}(c^{\mu(n)}_i)\otimes {V}_{t^{ n}}(E_{i}^{n}(~\mid j))\right)\\
& = H(R_{\mathrm{uni}})+H(Y_{\mathrm{uni}})+ \frac{1}{J_n}\frac{1}{n^3}\sum_{j=1}^{J_n}\sum_{i=1}^{n^3}S\left(
{V}_{t^{\mu(n)}}(c^{\mu(n)}_i)\otimes {V}_{t^{ n}}(E_{i}^{n}(~\mid j))\right)\text{ ,}\end{align*}\vspace{0.15cm}

By strong subadditivity of von Neumann entropy it holds $S(\varphi^{\mathfrak{J}H^{\mu(n)+n}}) + S(\varphi^{\mathfrak{H}H^{\mu(n)+n}})$
$\geq$ $S(\varphi^{H^{\mu(n)+n}})+S(\varphi^{\mathfrak{J}\mathfrak{H}H^{\mu(n)+n}})$. Thus by (\ref{secudet1}) we have

\begin{equation}\chi\left(R_{\mathrm{uni}},Z_{\mathcal{C}^{\mathrm{det}},t^{\mu(n) +n}}\right) \leq \lambda\text{ .}\label{secudet}\end{equation} \vspace{0.2cm}

\subsubsection{The legal receiver is able to decode the message}

We now use Theorem \ref{dichpart}. \ref{dichpart2} to show that the legal receiver's average error goes to
zero.\vspace{0.15cm}

For any $t^{\mu(n) +n}\in\theta^{\mu(n) +n}$, by (\ref{anf1}) and (\ref{anf2}),
\begin{align}&  P_e(\mathcal{C}^{\mathrm{det}}, t^{\mu(n) +n})\notag\\
&=1-
\frac{1}{J_n}\sum_{j=1}^{J_n}\mathrm{tr}\biggl(\left[\frac{1}{n^3}\sum_{i=1}^{n^3}
{U}_{t^{\mu(n)}}(c^{\mu(n)}_i)\otimes {U}_{t^{n}}(E_{i}^{n}(~\mid j))\right]\cdot\left[\sum_{k=1}^{n^3} D_k^{\mu(n)}\otimes
 D_{k,j}^{n}\right]
\biggr)\notag\\
&\leq 1-
\frac{1}{J_n}\sum_{j=1}^{J_n}\mathrm{tr}\biggl(\frac{1}{n^3}\sum_{i=1}^{n^3}\left[
{U}_{t^{\mu(n)}}(c^{\mu(n)}_i)\otimes {U}_{t^{n}}(E_{i}^{n}(~\mid j))\right]\cdot\left[ D_k^{\mu(n)}\otimes
 D_{k,j}^{n}\right]
\biggr)\notag\\
&= 1-
\frac{1}{J_n}\sum_{j=1}^{J_n}\mathrm{tr}\biggl(\frac{1}{n^3}\sum_{i=1}^{n^3}\left[{U}_{t^{\mu(n)}}(c^{\mu(n)}_i) D_k^{\mu(n)}
\right]\otimes\left[{U}_{t^{n}}(E_{i}^{n}(~\mid j))
 D_{k,j}^{n}\right]
\biggr)\notag\\
&= 1-\frac{1}{n^3}\sum_{i=1}^{n^3}\biggl(\mathrm{tr}\left[{U}_{t^{\mu(n)}}(c^{\mu(n)}_i) D_k^{\mu(n)}
\right]\cdot
\frac{1}{J_n}\sum_{j=1}^{J_n}\mathrm{tr}\left[{U}_{t^{n}}(E_{i}^{n}(~\mid j))
 D_{k,j}^{n}\right]\biggr)\notag\\
&= 1-\frac{1}{n^3}\sum_{i=1}^{n^3}\biggl(\mathrm{tr}\left[{U}_{t^{\mu(n)}}(c^{\mu(n)}_i) D_k^{\mu(n)}
\right]\cdot(1- P_e(C_i,t^n))\biggr)\notag\\
&\leq 1-(1-\vartheta -\lambda)\notag\\
&= \lambda+\vartheta\text{ ,}\label{errordet}\end{align}
the second inequality holds because for non-negative numbers
$\{\alpha_i,\beta_i:i=1,\ldots,M\}$ such that $\frac{1}{M}\sum_{i=1}^{M} \alpha_i$ $\leq$ $\vartheta$
and  $\frac{1}{M}\sum_{i=1}^{M} \beta_i$ $\leq$ $\lambda$ we have $\frac{1}{M}\sum_{i=1}^{M} (1-\alpha_i)(1-\beta_i)$ 
$\geq$ $1-\vartheta-\lambda$.
\vspace{0.2cm}

For any $n\in \mathbb{N}$ and positive $\lambda$, if there is an   $(n, J_n)$  randomness-assisted
 code $(\{\mathcal{C}^{\gamma}:\gamma\in \Lambda\},G)$ for
$\{(W_t,{V}_t): t \in \theta\}$ such that
\[ \max_{t^n\in\theta^n}\int_{\Lambda}P_{e}(\mathcal{C}^{\gamma},t^n)d
G(\gamma) < \epsilon \text{ ,}\]
\[\max_{t^n\in\theta^n} \int_{\Lambda}
\chi\left(R_{\mathrm{uni}},Z_{\mathcal{C}^{\gamma},t^n}\right)dG(\gamma) < \zeta\text{
,}\] choose $\delta  = \min \{\epsilon, \zeta  \} + \vartheta$,  by (\ref{errordet}) and  (\ref{secudet}),
we can find a $(\mu(n) +n, J_n)$
deterministic code $\mathcal{C}^{\mathrm{det}}=\biggl(E^{\mu(n)+n},\{D_{j}^{\mu(n)+n} :
j=1,\ldots,J_n\}\biggr)$ such that
such that
\[ \max_{t^{\mu(n) +n}\in\theta^{\mu(n) +n}}P_{e}(\mathcal{C}^{\mathrm{det}},t^{\mu(n) +n})< \lambda\text{ ,}\]
\[\max_{t^{\mu(n) +n}\in\theta^{\mu(n) +n}}
\chi\left(R_{\mathrm{uni}},Z_{\mathcal{C}^{\mathrm{det}},t^{\mu(n) +n}}\right) < \lambda\text{
.}\]

We know that $2^{\mu(n)}=o(n)$.
For any positive $\varepsilon$, if $n$ is large enough we have $\frac{1}{n}\log J_n -\frac{1}{\log n +n}\log J_n
\leq \varepsilon$.
Therefore, if the
arbitrarily varying classical-quantum  channel $\{W_{t}:
x\in\mathcal{X}\}$ is  not symmetrizable, we have
\begin{equation} C_s(\{(W_t,{V}_t): t \in \theta\};cr)\geq C_s(\{(W_t,{V}_t): t \in \theta\};r)- \varepsilon\text{ .}\label{varepsilon}\end{equation}
This and the fact that
\[C_s(\{(W_t,{V}_t): t \in \theta\};cr)\leq C_s(\{(W_t,{V}_t): t \in \theta\};r)\]
prove  Theorem \ref{dichpart}. \ref{dichpart1a}  (c.f. \cite{Ahl/Bli} for Ahlswede
dichotomy for arbitrarily varying classical-quantum channel
Channels). \vspace{0.2cm}

 \subsection{The proof of Theorem \ref{dichpart}. \ref{dichpart1b}}

 If  $\{W_t : t \in \theta\}$ is
symmetrizable,  the  deterministic capacity of
 $\{W_t : t \in \theta\}$ using a deterministic encoder is equal to zero by \cite{Ahl/Bli}. Now we
 have to check whether $C_s(\{(W_t,{V}_t): t \in \theta\})$ using stochastic encoder remains
  equal to zero. The proof is rather standard. Readers with
	experiences in information theory may pass over this subsection. \vspace{0.2cm}

For any $n\in\mathbb{N}$ and $J_n\in\mathbb{N}\setminus\{1\}$,
let $\mathcal{C} = \Bigl(E^n,\{D_j^n : j\in\{ 1,\ldots J_n\}\}\Bigr)$ be an  $(n, J_n)$  deterministic
code with a random encoder. We denote the set of all   deterministic encoders by 
$\mathbf{F}_n:=\Bigl\{f_n:\{1,\ldots,J_n\}\rightarrow {\mathbf{A}}^n \Bigr\}$.
 Since  the  deterministic capacity of
 $\{W_t : t \in \theta\}$ using deterministic encoder is zero, there is a positive $c$
such that for any $n\in\mathbb{N}$ we have
\begin{equation} \max_{t^n\in\theta^n}\frac{1}{J_n}\sum_{j=1}^{J_n}
\mathrm{tr}\biggl(W_{t^n}(f_n(j))D_j^n\biggr)<1-c\text{ .}\label{foranyn}\end{equation}
For any $t^n\in\theta^n$, we have
\begin{align}&1-c \notag\\
&=(1-c) \sum_{f_n\in\mathbf{F}_n}\prod_{k=1}^{J_n}E^n(f_n(k)\mid k) \notag\\
&>\sum_{f_n\in\mathbf{F}_n}\prod_{k=1}^{J_n}E^n(f_n(k)\mid k)
\frac{1}{J_n}\sum_{j=1}^{J_n} \mathrm{tr}\biggl(W_{t^n}(f_n(j))D_j^n\biggr)\notag\\
&=\frac{1}{J_n}\sum_{j=1}^{J_n}\sum_{a^n\in {\mathbf{A}}^n}  E^n(a^n\mid j)\mathrm{tr}\biggl(W_{t^n}(a^n)D_j^n\biggr)\notag\\
&=\frac{1}{J_n}\sum_{j=1}^{J_n}\mathrm{tr}\biggl( W_{t^n}( E^{n}(~\mid j))D_j^n\biggr)
\text{ ,}\label{detrand}\end{align}
the first equation holds because
\begin{align*}\allowdisplaybreaks[2]&\sum_{f_n\in\mathbf{F}_n}\prod_{j=1}^{J_n}E^n(f_n(j)\mid j)\\
&=\sum_{a^n}\sum_{f_n(1)=a^n}\biggl(\sum_{a^n}\sum_{f_n(2)=a^n}\biggl(\ldots
\biggl(\sum_{a^n}\sum_{f_n(J_n-1)=a^n}\\
&\biggl(\sum_{a^n}\sum_{f_n(J_n)=a^n}
\prod_{j=1}^{J_n}E^n(f_n(j)\mid j) \biggr)\biggr)\ldots\biggr)\biggr)\\
&=\sum_{a^n}\sum_{f_n(1)=a^n}\biggl(\sum_{a^n}\sum_{f_n(2)=a^n}\biggl(\ldots
\biggl(\sum_{a^n}\sum_{f_n(J_n-1)=a^n}\\
&\biggl(\sum_{a^n}
E^n(a^n\mid J_n)\prod_{j=1}^{J_n-1}E^n(f_n(j)\mid j) \biggr)\biggr)\ldots\biggr)\biggr)\\
&=\sum_{a^n}\sum_{f_n(1)=a^n}\left(\sum_{a^n}\sum_{f_n(2)=a^n}\left(\ldots\left(\sum_{a^n}\sum_{f_n(J_n-1)=a^n}
\prod_{j=1}^{J_n-1}E^n(f_n(j)\mid j) \right)\ldots\right)\right)\\
&=\sum_{a^n}\sum_{f_n(1)=a^n}\left(\sum_{a^n}\sum_{f_n(2)=a^n}\left(\ldots\left(\sum_{a^n}
E^n(a^n\mid J_n)\prod_{j=1}^{J_n-1}E^n(f_n(j)\mid j) \right)\ldots\right)\right)\\
&=\ldots\\
&=\sum_{a^n}\sum_{f_n(1)=a^n}E^n(f_n(1)\mid 1)\\
&=\sum_{a^n}E^n(a^n\mid 1)\\
&=1\text{ ,}\end{align*}
the second equation holds because
for any $j\in\{1,\ldots,J_n\}$, we have
\begin{align*}\allowdisplaybreaks[2]&
\sum_{f_n\in\mathbf{F}_n}\prod_{k=1}^{J_n}E^n(f_n(k)\mid k)\mathrm{tr}\biggl(W_{t^n}(f_n(j))D_j^n\biggr)\\
&=\sum_{a^n}\sum_{f_n(j)=a^n}E^n(a^n\mid j)\left(\prod_{k\not=j}E^n(f_n(k)\mid k)\right)\mathrm{tr}\biggl(W_{t^n}(f_n(j))D_j^n\biggr)\\
&=\sum_{a^n}\sum_{f_n(j)=a^n}E^n(a^n\mid j)\mathrm{tr}\biggl(W_{t^n}(f_n(j))D_j^n\biggr)\\
&=\sum_{a^n}E^n(a^n\mid j)\mathrm{tr}\biggl(W_{t^n}(a^n)D_j^n\biggr)
\text{ .}\end{align*}

By (\ref{detrand}), for   any  $n\in\mathbb{N}$, $J_n\in\mathbb{N}\setminus\{1\}$,
let $\mathcal{C}$ be any
 $(n, J_n)$
 deterministic  code with a random encoder,  if $\{W_t : t \in \theta\}$ is
symmetrizable, we have \[\max_{t\in\theta}P_e(\mathcal{C},t^n)>c\text{ .}\]
Thus the only achievable  deterministic secrecy capacity of
$\{(W_t,{V}_t): t \in \theta\}$ is $\log 1 =0$. Therefore
$C_s(\{(W_t,{V}_t): t \in \theta\})=0$. (Actually, (\ref{detrand})
shows that if $\{W_t : t \in \theta\}$ is symmetrizable, even the
deterministic capacity for message transmission of $\{(W_t,{V}_t): t \in \theta\}$
with random encoding technique
is equal to zero. Since the deterministic secrecy capacity
$C_s(\{(W_t,{V}_t): t \in \theta\})$ cannot exceed the  deterministic
 capacity for message transmission, we have $C_s(\{(W_t,{V}_t): t \in \theta\})=0$.)
This completes the proof of Theorem \ref{dichpart}. \ref{dichpart1b}.
\end{proof}

As we learn from  Example \ref{nachhinten}, there are indeed arbitrarily varying classical-quantum
 wiretap channels which have zero deterministic secrecy capacity and positive random  secrecy capacity.
Therefore, as Theorem \ref{dichpart}. \ref{dichpart1} shows,   randomness  is indeed a very helpful resource
 for the secure message transmission through an arbitrarily varying classical-quantum
 wiretap channel. But the problem is:  how  should the sender
and the receiver know  which  code is used in the particular
transmission?

Theorem \ref{dichpart}. \ref{dichpart2} shows  that  common randomness capacity is
always equal to the random secrecy capacity, even  for the arbitrarily varying classical-quantum
 wiretap channels of  Example \ref{nachhinten}. Therefore, common randomness  is
an equally helpful  resource
for the secure message transmission through an arbitrarily varying classical-quantum
 wiretap channel. However, as \cite{Bo/No} showed,  common randomness  is  a very ``costly'' resource.
 As Theorem \ref{dichpart} shows, for the transmission of  common randomness
we have to require that the
 deterministic capacity for message transmission of the sender's and legal receiver's
 channel is positive.
In the following Section \ref{AVCQWCWCA}, we will see that the much
``cheaper'' resource, the $m-a-(X,Y)$
  correlation, is also an equally helpful
    resource  for  the message transmission through an arbitrarily varying classical-quantum
channel. The advantage here is that we do not have to require that the
 deterministic capacity for message transmission of the sender's and legal receiver's
 channel is positive.

\section{Arbitrarily Varying Classical-Quantum Wiretap Channel with Correlation Assistance}\label{AVCQWCWCA}

In this section we consider the $m-a-(X,Y)$
 correlation-assisted  secrecy capacity of an arbitrarily
varying  classical-quantum wiretap channel.\vspace{0.15cm}

  Theorem \ref{dichpart}. \ref{dichpart2} shows that    common randomness is a helpful
    resource  for  the secure message transmission through an arbitrarily varying classical-quantum
 wiretap channel. The $m-a-(X,Y)$
  correlation is a weaker resource than  
 common randomness (cf. \cite{Bo/No}). We can simulate any $m-a-(X,Y)$
  correlation  by  common randomness asymptotically, but there exists a class of
sequences of bipartite distributions which cannot model   common randomness
(cf. Lemma 1 of \cite{Bo/No}).
However,  the results of \cite{Bo/No} show that the ``cheaper'' $m-a-(X,Y)$
  correlation is nevertheless a helpful
    resource   for    message transmission through an arbitrarily varying classical-quantum
channel. Our following Theorem \ref{causalcommonrandom} shows that also in  case of   secure message transmission through
an arbitrarily varying classical-quantum
 wiretap channel, the $m-a-(X,Y)$
  correlation  assistance  is  an  equally helpful
    resource  as    common randomness.

%\begin{theorem}
%Let $\{(W_t,{V}_t): t \in \theta\}$
% be an arbitrarily
%varying  classical-quantum wiretap channel. Let $\mathbf{X}$ and $\mathbf{Y}$ be finite sets.
%If $I(X,Y)>0$ holds for  a
%random variable $(X,Y)$ which is distributed to a joint probability distribution
%$p\in P(\mathbf{X}\times\mathbf{Y})$,
% then the randomness-assisted secrecy capacity is equal
%to the
% $m-a-(X,Y)$
% causal  correlation-assisted  secrecy capacity.
%\label{causalcommonrandom}
%\end{theorem}

\begin{theorem}
Let $\{(W_t,{V}_t): t \in \theta\}$
 be an arbitrarily
varying  classical-quantum wiretap channel. Let $\mathbf{X}$ and $\mathbf{Y}$ be finite sets.
If $I(X,Y)>0$ holds for  a
random variable $(X,Y)$ which is distributed according to a joint probability distribution
$p\in P(\mathbf{X}\times\mathbf{Y})$,
 then the randomness-assisted secrecy capacity is equal
to the
 $m-a-(X,Y)$
  correlation-assisted  secrecy capacity.
\label{causalcommonrandom}
\end{theorem}

\begin{proof}

Our proof is similar to the capacity results of arbitrarily varying  channels with correlation assistance in
\cite{Ahl/Cai} and \cite{Bo/No}. 

\subsection{When the  randomness-assisted code has positive secrecy capacity}

If the  randomness-assisted secrecy capacity of  $(W_t,{V}_t)_{t \in \theta}$ is positive,
we can build a new arbitrarily
varying  classical-quantum  channel
 $\{\tilde{U}_{t}:t \in \theta\}$ to
create  common randomness for the sender and the legal receiver. We show that this channel does not
have to be secure to be useful for a secure code for the 
original arbitrarily
varying  classical-quantum wiretap channel. Then, similar to our proof of
 Theorem \ref{dichpart}. \ref{dichpart1}, the sender and the 
 legal receiver can build two-part code word,
which consists of a non-secure code word for $\{\tilde{U}_{t}:t \in \theta\}$
to pass the index and
 a   common randomness-assisted secure code to transmit the message .\vspace{0.2cm}

At first we assume that the
 $m-a-(X,Y)$
      secrecy capacity of $\{(W_t,{V}_t): t \in \theta\}$ is positive, then
 the
 $m-a-(X,Y)$
      capacity of the arbitrarily varying classical-quantum channel
 $\{W_t : t \in \theta\}$ is positive. For the definition
of the capacity of an arbitrarily varying classical-quantum channel
please see \cite{Bo/No}.

By Theorem \ref{dichpart}. \ref{dichpart2},  the randomness-assisted secrecy
capacity is equal to the common randomness-assisted secrecy
capacity. Let $\delta>0$, $\zeta>0$, and  $\epsilon>0$, and
 $\biggl\{\mathcal{C}^{\gamma}=\Bigl(E_{\gamma}^n,\{D_{\gamma,j}^n:
 j\in \{1,\ldots J_n\}\}\Bigr):\gamma\in \Gamma\biggr\}$ be an $(n, J_n)$  common randomness-assisted
quantum code such that $\frac{\log J_n}{n} > C_{s}((W_t,{V}_t)_{t
\in \theta},r)-\delta$, and
\[ \max_{t^n\in\theta^n} \frac{1}{\left|\Gamma\right|} \sum_{\gamma=1}^{\left|\Gamma\right|}P_{e}(\mathcal{C}^{\gamma},t^n) < \epsilon\text{ ,}\]
\[\max_{t^n\in\theta^n} \frac{1}{\left|\Gamma\right|} \sum_{\gamma=1}^{\left|\Gamma\right|}
\chi\left(R_{\mathrm{uni}},Z_{\mathcal{C}^{\gamma},t^n}\right) < \zeta\text{ .}\]\vspace{0.2cm}

We denote $\mathfrak{F} := \{ f: f\text{ is a function } \mathbf{X}\rightarrow \mathbf{A}\}$. 
Let $H_{\mathbf{Y}}$  be a Hilbert space of dimension $|\mathbf{Y}|$ and
$\{{ \breve{\kappa}}_{y}: y\in \mathbf{Y}\}$ be a set of  pairwise orthogonal and pure states on  $H_{\mathbf{Y}}$.
For every $t\in\theta$,
\begin{equation} \tilde{U}_{t}(f) :=
\sum_{\mathbf{x}} \sum_{\mathbf{y}}
 p(\mathbf{x},\mathbf{y}){ \breve{\kappa}}_{y} \otimes W_{t}\left(f(\mathbf{x})\right) \label{newAVC}
\end{equation} 
%where for $f^n = (f_1,\ldots, f_n) \in \mathfrak{F}^n$ and $\mathbf{x} = (\mathbf{x}_1,\ldots, \mathbf{x}_n) \in \mathbf{X}^n$,
%$f^n(\mathbf{x}^n) = (f_1(\mathbf{x}_1),\ldots, f_{n}(\mathbf{x}_n))$,
 defines a
 classical-quantum  channel \[\tilde{U}_{t}: \mathfrak{F}  \rightarrow \mathcal{S}(H\otimes H_{\mathbf{Y}})\text{ .}\]
$\{\tilde{U}_{t}:t \in \theta\}$
defines an
arbitrarily
varying  classical-quantum  channel $ \mathfrak{F}  \rightarrow \mathcal{S}(H)\otimes H_{\mathbf{Y}}$.\vspace{0.15cm}

In \cite{Bo/No} (see also \cite{Ahl/Cai} for a classical version),
it was shown that if $I(X,Y)$ is positive, the deterministic
capacity of $(\tilde{U}_{t})_{t \in \theta}$ is equal to the
 $m-a-(X,Y)$
      capacity of
 $\{W_t : t \in \theta\}$.
 By Remark \ref{detvsran},
we may assume that the  deterministic
capacity of $(\tilde{U}_{t})_{t \in \theta}$
using deterministic encoder is positive.
This means that
the sender and the
receiver can build a code
$\Bigl(\left(f^{\nu(n)}_{\gamma}\right)_{\gamma=1,\ldots,\left|\Gamma\right|},   \{D_{\gamma}^{\nu(n)}:
\gamma=1,\ldots,\left|\Gamma\right|\}\Bigr)$ with deterministic encoder
for $(\tilde{U}_{t})_{t \in \theta}$ of
length $\nu(n)$, where  $2^{\nu(n)}$ is in polynomial order of $n$
and $f^{\nu(n)}_{\gamma} (\mathbf{x}^{\nu(n)}) = \Big(f_{\gamma,1}(\mathbf{x}_1),\ldots,f_{\gamma,\nu(n)}
(\mathbf{x}_{\nu(n)})\Big)$ for $\mathbf{x}^{\nu(n)} = (\mathbf{x}_1,\ldots,\mathbf{x}_{\nu(n)})$, such
that the following statement is valid. For any positive $\vartheta$, if $n$ is large enough, we have
\begin{align} &1-\vartheta\notag\\
&\leq\min_{t^{\nu(n)}\in\theta^{\nu(n)}} \frac{1}{\left|\Gamma\right|}
 \sum_{\gamma=1}^{\left|\Gamma\right|}
\mathrm{tr}\biggl(\tilde{U}_{t^{\nu(n)}}\left(f^{\nu(n)}_{\gamma}\right)
D_{\gamma}^{\nu(n)}\biggr)\notag\\
&=\min_{t^{\nu(n)}\in\theta^{\nu(n)}} \frac{1}{\left|\Gamma\right|}\mathrm{tr}\biggl(
 \sum_{\gamma=1}^{\left|\Gamma\right|} \sum_{\mathbf{x}^{\nu(n)}\in\mathbf{X}^{\nu(n)}}
 \sum_{\mathbf{y}^{\nu(n)}\in\mathbf{Y}^{\nu(n)}} p\left(\mathbf{x}^{\nu(n)},\mathbf{y}^{\nu(n)}\right)\notag\\
&\cdot \Bigl[{ \breve{\kappa}}_{y^n} \otimes W_{t^{\nu(n)}}(f^{\nu(n)}_{\gamma}(\mathbf{x}^{\nu(n)}))\Bigr]  D_{\gamma}^{\nu(n)}\biggr)\notag\\
 &=\min_{t^{\nu(n)}\in\theta^{\nu(n)}} \frac{1}{\left|\Gamma\right|}
 \sum_{\gamma=1}^{\left|\Gamma\right|} \sum_{\mathbf{x}^{\nu(n)}\in\mathbf{X}^{\nu(n)}}
\sum_{\mathbf{y}^{\nu(n)}\in\mathbf{Y}^{\nu(n)}} p\left(\mathbf{x}^{\nu(n)},\mathbf{y}^{\nu(n)}\right)\notag\\
 &\cdot \mathrm{tr}\biggl(W_{t^{\nu(n)}}(c^{\nu(n)}_{\mathbf{x}^{\nu(n)},\gamma})
 D_{(\mathbf{y}^{\nu(n)}),\gamma}^{\nu(n)}\biggr)\text{ ,} \label{ineqfornocausal}\end{align}
where for every $\gamma\in\Gamma$, $\mathbf{x}^{\nu(n)}= (\mathbf{x}_1,\ldots,\mathbf{x}_{{\nu(n)}})\in\mathbf{X}^{\nu(n)}$,
 and $\mathbf{y}^{\nu(n)}= (\mathbf{y}_1,\ldots,\mathbf{y}_{{\nu(n)}})\in\mathbf{Y}^{\nu(n)}$,
 we set  $p(\mathbf{x}^{\nu(n)},\mathbf{y}^{\nu(n)})=\prod_{i,j} p(\mathbf{x}_i,\mathbf{y}_j)$,
\[c^{\nu(n)}_{\mathbf{x}^{\nu(n)},\gamma}:=f^{\nu(n)}_{\gamma}(\mathbf{x}^{\nu(n)})\in {\mathbf{A}}^{\nu(n)}\text{ ,} \]
  and
    \[  D_{(\mathbf{y}^{\nu(n)}),\gamma}^{\nu(n)}:=  \mathrm{tr}_{H_{\mathbf{Y}^{\nu(n)}}}   \left(
    ({ \breve{\kappa}}_{y^{\nu(n)}} \otimes \mathrm{id}_{H^{\otimes \nu(n)}})  D_{\gamma}^{\nu(n)}  \right)     \text{ .} \]
    The last equation of (\ref{ineqfornocausal}) holds because
\begin{align*} & \mathrm{tr}\biggl(W_{t^{\nu(n)}}(c^{\nu(n)}_{\mathbf{x}^{\nu(n)},\gamma})
    \mathrm{tr}_{H_{\mathbf{Y}^{\nu(n)}}}   \left(
({ \breve{\kappa}}_{y^{\nu(n)}} \otimes \mathrm{id}_{H^{\otimes \nu(n)}})  D_{\gamma}^{\nu(n)}  \right)  \biggr)\\
    &=  \mathrm{tr}\left(\Bigl[\mathrm{id}_{H_{\mathbf{Y}^{\nu(n)}}} \otimes  W_{t^{\nu(n)}}(c^{\nu(n)}_{\mathbf{x}^{\nu(n)},\gamma})\Bigr]
    \Bigl[{ \breve{\kappa}}_{y^{\nu(n)}} \otimes \mathrm{id}_{H^{\otimes \nu(n)}}\Bigr]  D_{\gamma}^{\nu(n)}     \right)\\
    &=\mathrm{tr}\biggl(
    \Bigl[{ \breve{\kappa}}_{y^{\nu(n)}} \otimes  W_{t^{\nu(n)}}(c^{\nu(n)}_{\mathbf{x}^{\nu(n)},\gamma})\Bigr]  D_{\gamma}^{\nu(n)}     \biggr)
    \text{ .} \end{align*}
    \vspace{0.15cm}

Since $\sum_{\gamma=1}^{\left|\Gamma\right|} D_{(\mathbf{y}^{\nu(n)}),\gamma}^{\nu(n)}$
$=$ $\sum_{\gamma=1}^{\left|\Gamma\right|}\mathrm{tr}_{H_{\mathbf{Y}^{\nu(n)}}}   \Bigl(
    ({ \breve{\kappa}}_{y^{\nu(n)}} \otimes \mathrm{id}_{H^{\otimes \nu(n)}})  D_{\gamma}^{\nu(n)}  \Bigr)$
    $=$ $\mathrm{tr}_{H_{\mathbf{Y}^{\nu(n)}}}   \Bigl(
    ({ \breve{\kappa}}_{y^{\nu(n)}} \otimes \mathrm{id}_{H^{\otimes \nu(n)}})  \sum_{\gamma=1}^{\left|\Gamma\right|}D_{\gamma}^{\nu(n)} \Bigr) $
$= $ $\mathrm{id}_{{H}^{\otimes \nu(n)}}$,
we can define an   $(X,Y)$-correlation-assisted
  $(\nu(n), \left|\Gamma\right|)$  code  (this is a  code with deterministic encoder)
by $\biggl(\Bigl(c^{\nu(n)}_{\mathbf{x}^{\nu(n)},\gamma}\Bigr)_{\gamma\in\{1,\ldots,\left|\Gamma\right|\}},$ $\{D_{(\mathbf{y}^{\nu(n)}),\gamma}^{\nu(n)}:$
$\gamma\in\{1,\ldots,\left|\Gamma\right|\}\}\biggr)$.\vspace{0.2cm}

%As in proof of Theorem \ref{dichpart}, we may assume that this code is not secure against the wiretapper, i.e.
%\begin{equation}\min_{t^{\nu(n)}\in\theta^{\nu(n)}}\sum_{\mathbf{x}^{\nu(n)}\in\mathbf{X}^{\nu(n)}}
% \sum_{\mathbf{y}^{\nu(n)}\in\mathbf{Y}^{\nu(n)}}p(\mathbf{x}^{\nu(n)},\mathbf{y}^{\nu(n)}) I(Y_{\mathrm{uni}},
%\mathcal{P}_{t^{\nu(n)},\mathbf{x}^{\nu(n)}})\geq  H(Y_{\mathrm{uni}}) -\delta\text{
%,}\label{codet2b}
%\end{equation} where $Y_{\mathrm{uni}}$ is the
%random variable uniformly distributed  on $\{1,\ldots
%,\left|\Gamma\right|\}$, and $\mathcal{P}_{t^{\nu(n)},\mathbf{x}^{\nu(n)}}:=
%\left\{{V}_{t^{\nu(n)}}({c}^{\nu(n)}_{\mathbf{x}^{\nu(n)},1}),{V}_{t^{\nu(n)}}({c}^{\nu(n)}_{\mathbf{x}^{\nu(n)},2}),\ldots,
%{V}_{t^{\nu(n)}}({c}^{\nu(n)}_{\mathbf{x}^{\nu(n)},{\nu(n)}})\right\}$ is the
%wiretapper's resulting quantum state at channel state $t^{\nu(n)}$  when
%$X^{\nu(n)}=\mathbf{x}^{\nu(n)}$. In the other case, the sender and the receiver
%can use
%$\biggl\{\Bigl(c^{\nu(n)}_{\mathbf{x}^{\nu(n)},i},\{D_{(\mathbf{y}^{\nu(n)}),\gamma}^{\nu(n)}: \gamma\in\{1,\ldots,\left|\Gamma\right|\}\}\Bigr)
%:\mathbf{x}^{\nu(n)}\in\mathbf{X}^{\nu(n)},\mathbf{y}^{\nu(n)}\in\mathbf{Y}^{\nu(n)}\biggr\}$ to
%create a secret key (c.f. \cite{Ahl/Csi1}, \cite{Ahl/Csi2} and \cite{De/Win}) for
%every $\mathbf{x}^{\nu(n)}\in\mathbf{X}^{\nu(n)}$ and
%$\mathbf{y}^{\nu(n)}\in\mathbf{Y}^{\nu(n)}$ to get a even larger capacity than the
%deterministic secrecy capacity without a secret key.

Now we can construct an $(X,Y)$-correlation-assisted $(\nu(n) +n,
J_n)$ code $\mathcal{C}(X,Y)$ $=$
$\biggl\{\Bigl(E_{\mathbf{x}^{\nu(n)+n}},\{D_j^{\mathbf{y}^{\nu(n)+n}}:
j\in \{ 1,\ldots
,J_n\}\}\Bigr) :\mathbf{x}^{\nu(n)+n}\in\mathbf{X}^{\nu(n)+n},\mathbf{y}^{\nu(n)+n}\in\mathbf{Y}^{\nu(n)+n}\biggr\}$,
where for $\mathbf{x}^{\nu(n)
+n}=(\mathbf{x}^{\nu(n)},\mathbf{x}^{n})$,  $\mathbf{y}^{\nu(n)
+n}=(\mathbf{y}^{\nu(n)},\mathbf{y}^{n})$ and $a^{\nu(n)
+n}=(a^{\nu(n)},a^{n})\in {\mathbf{A}}^{\nu(n) +n}$
\[E_{\mathbf{x}^{\nu(n)+n}}(a^{\nu(n) +n}|j)=\begin{cases}
  \frac{1}{\left|\Gamma\right|}E_{\gamma}(a^{n}|j) \text{ if } a^{\nu(n)} = c^{\nu(n)}_{\mathbf{x}^{\nu(n)},\gamma}\\
    0 \text{ else}             \end{cases}\text{ ,}
\] and \[D_j^{\mathbf{y}^{\nu(n)+n}} := \sum_{\gamma=1}^{\left|\Gamma\right|}
D_{(\mathbf{y}^{\nu(n)}),\gamma}^{\nu(n)}\otimes D_{\gamma,j}^{n} \text{ .}
\]

% $\biggl(\left(c^{\nu(n)}_{\mathbf{x}^{\nu(n)},\gamma}\right)_{\gamma\in\{1,\ldots,\left|\Gamma\right|\}}, \{D_{(\mathbf{y}^{\nu(n)}),\gamma}^{\nu(n)}:
%\gamma\in\{1,\ldots,\left|\Gamma\right|\}\}\biggr)$ is causal. Since $\mathcal{C}^{\gamma}$ do
%not depend on the outputs of $X$ for all $\gamma\in \Gamma$, it is also causal.
%The encoder of
%$\mathcal{C}(X,Y)$ is a composition of $\biggl(\left(c^{\nu(n)}_{\mathbf{x}^{\nu(n)},\gamma}\right)_{\gamma\in\{1,\ldots,\left|\Gamma\right|\}}, \{D_{(\mathbf{y}^n),\gamma}^{\nu(n)}:
%\gamma\in\{1,\ldots,\left|\Gamma\right|\}\}\biggr)$ and $\mathcal{C}^{\gamma}$, thus
%$\mathcal{C}(X,Y)$ is causal.\vspace{0.15cm}

For any $\gamma\in \{1,\ldots ,\left|\Gamma\right|\}$ let
\begin{align*}&\mathfrak{Z}_{\gamma,t^{\nu(n)+n},\mathbf{x}^{\nu(n)+n}}\\
& :=\biggl\{ {V}_{t^{\nu(n)}}
\left(c^{\nu(n)}_{\mathbf{x}^{\nu(n)},\gamma}\right)\otimes {V}_{t^{n}}  \left(E_{\gamma}(a^{n}|1)\right),\ldots,\\
& {V}_{t^{\nu(n)}}
\left(c^{\nu(n)}_{\mathbf{x}^{\nu(n)},\gamma}\right)\otimes {V}_{t^{n}}  \left(E_{\gamma}(a^{n}|J_n)\right)\biggr\}\text{ .}\end{align*} 
 Similar to (\ref{einfach}), for any
$\mathbf{x}^{\nu(n)+n}\in\mathbf{X}^{\nu(n)+n}$, $\gamma\in\Gamma$, and $t^{\nu(n)
+n}= (t^{\nu(n) }t^{n})$  we have
\begin{align}& \chi\left(R_{\mathrm{uni}},\mathfrak{Z}_{\gamma,t^{\nu(n)+n},\mathbf{x}^{\nu(n)+n}}\right)=
\chi\left(R_{\mathrm{uni}},Z_{\mathcal{C}^{\gamma},t^n}\right)\text{
.}\label{einfachb}\end{align}

By definition we have\begin{align*}&Z_{t^{\nu(n)+n},\mathbf{x}^{\nu(n)+n}} :=\\
&\biggl\{\frac{1}{\left|\Gamma\right|}
\sum_{\gamma=1}^{\left|\Gamma\right|}{V}_{t^{\nu(n)}}(c^{\nu(n)}_{\mathbf{x}^{\nu(n)},\gamma})\otimes{V}_{t^{n}}(
E_{\gamma}(~|1)), \ldots,\\
& \frac{1}{\left|\Gamma\right|}
\sum_{i=1}^{\left|\Gamma\right|}{V}_{t^{\nu(n) }}(c^{\nu(n)}_{\mathbf{x}^{\nu(n)},\gamma})\otimes{V}_{t^{n}}(
E_{\gamma}(~|J_n))\biggr\}\text{ .}\end{align*} 

%For $t^{\nu(n) +n}= (t^{\nu(n) },t^{n})$, $\mathbf{x}^{\nu(n) +n}=(\mathbf{x}^{\nu(n)},\mathbf{x}^{n})$ we denote
%\begin{align*}&\mathcal{Q}_{t^{\nu(n)+n},\mathbf{x}^{\nu(n)+n}} :=\\
%&\biggl\{\frac{1}{J_n}
% \sum_{j=1}^{J_n} \sum_{a^n\in {\mathbf{A}}^n}E_{1}(a^{n}|j){V}_{t^{\nu(n) }}(c^{\nu(n)}_{\mathbf{x}^{\nu(n)},1})\otimes{V}_{t^{n}}(
%a^n),\ldots,\\
%& \frac{1}{J_n}
% \sum_{j=1}^{J_n} \sum_{a^n\in {\mathbf{A}}^n} E_{\left|\Gamma\right|}(a^{n}|j){V}_{t^{\nu(n)}}(c^{\nu(n)}_{\mathbf{x}^{\nu(n)},\left|\Gamma\right|})\otimes{V}_{t^{n}}(
%a^n)\biggr\}\text{ .}\end{align*}
%For every $j\in\{1,\ldots,J_n\}$ we denote
%\begin{align*}&\mathcal{Q}_{j,t^{\nu(n)+n},\mathbf{x}^{\nu(n)+n}} :=\\
%&\biggl\{ \sum_{a^n\in {\mathbf{A}}^n}E_{1}(a^{n}|j){V}_{t^{\nu(n)
%}}(c^{\nu(n)}_{\mathbf{x}^{\nu(n)},1})\otimes{V}_{t^{n}}(
%a^n),\ldots,\\
%& \sum_{a^n\in {\mathbf{A}}^n}E_{\left|\Gamma\right|}(a^{n}|j){V}_{t^{\nu(n)
%}}(c^{\nu(n)}_{\mathbf{x}^{\nu(n)},\left|\Gamma\right|})\otimes{V}_{t^{n}}(
%a^n)\biggr\}\text{ .}\end{align*}

Similar to (\ref{secudet1})  let  $\lambda:=\min\{\epsilon, \zeta\}$,
 for any
$t^{\nu(n) +n}= (t^{\nu(n) },t^{n})$, $\mathbf{x}^{\nu(n) +n}=(\mathbf{x}^{\nu(n)},\mathbf{x}^{n})$
and  $\mathbf{y}^{\nu(n) +n}=(\mathbf{y}^{\nu(n)},\mathbf{y}^{n})$  we have

\begin{align}&\sum_{\mathbf{x}^{\nu(n) +n}\in\mathbf{X}^{\nu(n) +n}}\sum_{\mathbf{y}^{\nu(n) +n}\in\mathbf{Y}^{\nu(n) +n}}
p\left(\mathbf{x}^{\nu(n) +n},\mathbf{y}^{\nu(n) +n}\right)\chi\left(R_{\mathrm{uni}},Z_{t^{\nu(n)+n},\mathbf{x}^{\nu(n)+n}}\right)\notag\\
&\leq\sum_{\mathbf{x}^{\nu(n) +n}} \sum_{\mathbf{y}^{\nu(n) +n}}
p\left(\mathbf{x}^{\nu(n) +n},\mathbf{y}^{\nu(n) +n}\right)
\chi\left(R_{\mathrm{uni}},Z_{t^{\nu(n)+n},\mathbf{x}^{\nu(n)+n}}\right)\notag\\
&-\sum_{\mathbf{x}^{\nu(n)+n}}
 \sum_{\mathbf{y}^{\nu(n)+n}} p\left(\mathbf{x}^{\nu(n)+n},\mathbf{y}^{\nu(n)+n}\right)\frac{1}{\left|\Gamma\right|}\sum_{\gamma=1}^{\left|\Gamma\right|}
\chi\left(R_{\mathrm{uni}},Z_{\mathcal{C}^{\gamma},t^n}\right)+\lambda \notag\\
&=\sum_{\mathbf{x}^{\nu(n) +n}} \sum_{\mathbf{y}^{\nu(n) +n}}
p\left(\mathbf{x}^{\nu(n) +n},\mathbf{y}^{\nu(n) +n}\right)
\chi\left(R_{\mathrm{uni}},Z_{t^{\nu(n)+n},\mathbf{x}^{\nu(n)+n}}\right)\notag\\
&-\sum_{\mathbf{x}^{\nu(n)+n}}
 \sum_{\mathbf{y}^{\nu(n)+n}} p\left(\mathbf{x}^{\nu(n)+n},\mathbf{y}^{\nu(n)+n}\right)\frac{1}{\left|\Gamma\right|}\sum_{\gamma=1}^{\left|\Gamma\right|}
\chi\left(R_{\mathrm{uni}},\mathfrak{Z}_{\mathcal{C}^{\gamma},t^{\nu(n)+n},\mathbf{x}^{\nu(n)+n}}\right)+\lambda \notag\\
&= \sum_{\mathbf{x}^{\nu(n) +n}} \sum_{\mathbf{y}^{\nu(n) +n}}
p\left(\mathbf{x}^{\nu(n) +n},\mathbf{y}^{\nu(n) +n}\right)\Biggl[ S\left( \frac{1}{J_n}\frac{1}{\left|\Gamma\right|}\sum_{j=1}^{J_n}
\sum_{i=1}^{\left|\Gamma\right|}{V}_{t^{\nu(n) }}\left(c^{\nu(n)}_{\mathbf{x}^{\nu(n)},\gamma}\right)\otimes{V}_{t^{n}}\left(
E_{\gamma}(~|j)\right) \right)\notag\\
&-\frac{1}{J_n}\sum_{j=1}^{J_n}S\left( \frac{1}{\left|\Gamma\right|}
\sum_{i=1}^{\left|\Gamma\right|}{V}_{t^{\nu(n) }}\left(c^{\nu(n)}_{\mathbf{x}^{\nu(n)},\gamma}\right)\otimes{V}_{t^{n}}\left(
E_{\gamma}(~|j)\right) \right)\notag\\
&- \frac{1}{\left|\Gamma\right|}\sum_{i=1}^{\left|\Gamma\right|}S\left( \frac{1}{J_n}\sum_{j=1}^{J_n}
{V}_{t^{\nu(n) }}\left(c^{\nu(n)}_{\mathbf{x}^{\nu(n)},\gamma}\right)\otimes{V}_{t^{n}}\left(
E_{\gamma}(~|j)\right) \right)\notag\\
&+ \frac{1}{J_n}\frac{1}{\left|\Gamma\right|}\sum_{j=1}^{J_n}
\sum_{i=1}^{\left|\Gamma\right|}S\left( {V}_{t^{\nu(n) }}\left(c^{\nu(n)}_{\mathbf{x}^{\nu(n)},\gamma}\right)\otimes{V}_{t^{n}}\left(
E_{\gamma}(~|j)\right) \right)\Biggr]+\lambda\notag\\
 &\leq\lambda\text{ .}\label{secudetb}\end{align}

By (\ref{ineqfornocausal}), for any $t^{\nu(n) +n}\in\theta^{\nu(n) +n}$, 
\begin{align}& \sum_{\mathbf{x}^{\nu(n) +n}} \sum_{\mathbf{y}^{\nu(n) +n}}
p(\mathbf{x}^{\nu(n) +n},\mathbf{y}^{\nu(n) +n}) P_e\left(\mathcal{C}(\mathbf{x}^{\nu(n) +n},\mathbf{y}^{\nu(n) +n}), t^{\nu(n)+n}\right)\notag\\
&=1-\sum_{\mathbf{x}^{\nu(n) +n}} \sum_{\mathbf{y}^{\nu(n) +n}}
p\left(\mathbf{x}^{\nu(n) +n},\mathbf{y}^{\nu(n)
+n}\right)\frac{1}{J_n}\sum_{j=1}^{J_n}\mathrm{tr}\biggl(\biggl[\frac{1}{\left|\Gamma\right|}
 \sum_{\gamma=1}^{\left|\Gamma\right|}\notag\\
&{V}_{t^{\nu(n) }}(c^{\nu(n)}_{\mathbf{x}^{\nu(n)},\gamma})\otimes
 {V}_{t^{n}}(E_{\gamma}(~|j))\biggr]\cdot\left[\sum_{\gamma=1}^{\left|\Gamma\right|}
D_{(\mathbf{y}^{\nu(n)}),\gamma}^{\nu(n)}\otimes D_{\gamma,j}^{n}\right]
\biggr)\notag\\
&\leq 1-\sum_{\mathbf{x}^{\nu(n) +n}} \sum_{\mathbf{y}^{\nu(n) +n}}
p(\mathbf{x}^{\nu(n) +n},\mathbf{y}^{\nu(n)
+n})
\frac{1}{J_n}\sum_{j=1}^{J_n}\mathrm{tr}\biggl(\frac{1}{\left|\Gamma\right|}
 \sum_{\gamma=1}^{\left|\Gamma\right|}\notag\\
&\left[{V}_{t^{\nu(n) }}(c^{\nu(n)}_{\mathbf{x}^{\nu(n)},\gamma})\otimes
 {V}_{t^{n}}(E_{\gamma}(~|j))\right]\cdot\left[
D_{(\mathbf{y}^{\nu(n)}),\gamma}^{\nu(n)}\otimes D_{\gamma,j}^{n}\right]
\biggr)\notag\\
&=1-\sum_{\mathbf{x}^{\nu(n) }} \sum_{\mathbf{y}^{\nu(n) }}
p(\mathbf{x}^{\nu(n) },\mathbf{y}^{\nu(n)})
\frac{1}{J_n}\sum_{j=1}^{J_n}\mathrm{tr}\biggl(\frac{1}{\left|\Gamma\right|}
 \sum_{\gamma=1}^{\left|\Gamma\right|}\notag\\
&\left[{V}_{t^{\nu(n) }}(c^{\nu(n)}_{\mathbf{x}^{\nu(n)},\gamma})
 D_{(\mathbf{y}^{\nu(n)}),\gamma}^{\nu(n)}\right]\otimes\left[
{V}_{t^{n}}(E_{\gamma}(~|j)) D_{\gamma,j}^{n}\right]
\biggr)\notag\\
&=1-\sum_{\mathbf{x}^{\nu(n) }} \sum_{\mathbf{y}^{\nu(n) }}
p(\mathbf{x}^{\nu(n) },\mathbf{y}^{\nu(n)})\frac{1}{\left|\Gamma\right|}
 \sum_{\gamma=1}^{\left|\Gamma\right|}\mathrm{tr}\left(
{V}_{t^{\nu(n) }}(c^{\nu(n)}_{\mathbf{x}^{\nu(n)},\gamma})
 D_{(\mathbf{y}^{\nu(n)}),\gamma}^{\nu(n)}\right) \notag\\
&\cdot\left(\frac{1}{J_n}\sum_{j=1}^{J_n} \mathrm{tr}(
{V}_{t^{n}}(E_{\gamma}(~|j)) D_{\gamma,j}^{n})\right)\notag\\
&\leq \lambda+\vartheta\text{ .}\label{errordetb}\end{align}
\vspace{0.2cm}

We now combine (\ref{errordetb}) and  (\ref{secudetb}) and obtain the following result. 

If $I(X,Y)$ and
 the $m-a-(X,Y)$
      secrecy capacity of $\{(W_t,{V}_t): t \in \theta\}$ are positive, 
we define $\lambda := \min\{\epsilon, \zeta\} + \vartheta$
and the following
statement is valid.
For any $n\in \mathbb{N}$ and positive $\lambda$, if there is an $(n,
J_n)$ randomness-assisted
 code $(\{\mathcal{C}^{\gamma}:\gamma\in \Lambda\},G)$ for
$\{(W_t,{V}_t): t \in \theta\}$ such that
\[ \max_{t^n\in\theta^n}\int_{\Lambda}P_{e}(\mathcal{C}^{\gamma},t^n)d
G(\gamma) < \epsilon\text{ ,}\] and
\[\max_{t^n\in\theta^n} \int_{\Lambda}
\chi\left(R_{\mathrm{uni}},Z_{\mathcal{C}^{\gamma},t^n}\right)dG(\gamma) <
\zeta\text{ ,}\] then there is also a   $(\nu(n)+n, J_n)$ common
randomness-assisted
 code $\mathcal{C}(X,Y) =
\bigg\{\left(E_{\mathbf{x}^{\nu(n)+n}},D_j^{\mathbf{y}^{\nu(n)+n}}\right):
j\in \{ 1,\ldots
,J_n\},\mathbf{x}^{\nu(n)+n}\in\mathbf{X}^{\nu(n)+n}\mathbf{y}^{\nu(n)+n}\in\mathbf{Y}^{\nu(n)+n}\bigg\}$
such that
\begin{align}\label{ausf1} &\max_{t^{\nu(n)+n} \in \theta^{\nu(n)+n}} \sum_{\mathbf{x}^{\nu(n)+n}\in\mathbf{X}^{\nu(n)+n}}
\sum_{\mathbf{y}^{\nu(n)+n}\in\mathbf{Y}^{\nu(n)+n}}p(\mathbf{x}^{\nu(n)+n},
 \mathbf{y}^{\nu(n)+n})\notag\\
 & P_{e} (C(\mathbf{x}^{\nu(n)+n},\mathbf{y}^{\nu(n)+n}), t^{\nu(n)+n}) < \lambda\text{ ,}\end{align} and
\begin{equation} \label{ausf2}\max_{t^{\nu(n)+n}\in\theta^{\nu(n)+n}}  
\chi\left(R_{\mathrm{uni}};Z_{t^{\nu(n)+n},\mathbf{x}^{\nu(n)+n}}\mid X\right) <
\lambda\text{ .}\end{equation}

(\ref{ausf1}) and (\ref{ausf2}) mean that
\begin{align*}&C_s(\{(W_t,{V}_t): t \in \theta\};corr(X,Y))\geq C_s(\{(W_t,{V}_t): t \in \theta\};r)-\frac{1}{n}\cdot\log J_n\\
& +\frac{1}{\nu(n) +n}\log J_n  \text{ .}\end{align*}
We know that  $2^{\nu(n)}$ is
in polynomial order of $n$. For  any positive
$\varepsilon$, if $n$ is large enough we have $\frac{1}{n}\log
J_n -\frac{1}{\log n +n}\log J_n \leq \varepsilon$.
 Therefore, if $I(X,Y)$ and $C_s(\{(W_t,{V}_t): t \in \theta\};corr(X,Y))$ are both positive, we have
\begin{equation} C_s(\{(W_t,{V}_t): t \in \theta\};corr(X,Y))\geq C_s(\{(W_t,{V}_t): t \in \theta\};r)-\varepsilon \text{ .}\end{equation}
This and the fact that
\[C_s(\{(W_t,{V}_t): t \in \theta\};corr(X,Y))\leq C_s(\{(W_t,{V}_t): t \in \theta\};r) \text{ ,}\]
prove  Theorem \ref{causalcommonrandom} for the case that $C_s((W_t,{V}_t)_{t \in
\theta};corr(X,Y))$ is positive. 

\subsection{When the  randomness-assisted code has zero secrecy capacity}

If the  randomness-assisted secrecy capacity of  $(W_t,{V}_t)_{t \in \theta}$ is equal
to zero, with a similar technique as the techniques in \cite{Ahl/Cai} and \cite{Bo/No}
we   show that the $(X,Y)$
  correlation-assisted  secrecy capacity of $(W_t,{V}_t)_{t \in \theta}$
	is also equal
to zero.\vspace{0.2cm}

 Now we assume that the
 $m-a-(X,Y)$
      secrecy capacity of $\{(W_t,{V}_t): t \in \theta\}$ is
 equal to zero. If $C_s(\{(W_t,{V}_t): t \in \theta\};r)$ is also equal to
 zero, then there is nothing to prove. Thus let us assume that
  $C_s(\{(W_t,{V}_t): t \in \theta\};r)$ is positive.\vspace{0.15cm}

Assume that there is an $(n, J_n)$ randomness-assisted
 code $(\{\mathcal{C}^{\gamma}:\gamma\in \Lambda\},G)$ for
$\{(W_t,{V}_t): t \in \theta\}$ such that
\[ \max_{t^n\in\theta^n}\int_{\Lambda}P_{e}(\mathcal{C}^{\gamma},t^n)d
G(\gamma) < \lambda\text{ ,}\]
\[\max_{t^n\in\theta^n} \int_{\Lambda}
\chi\left(R_{\mathrm{uni}},Z_{\mathcal{C}^{\gamma},t^n}\right)dG(\gamma) <
\lambda\text{ .}\]
 We denote $\mathfrak{F}$ and the
arbitrarily varying  classical-quantum  channel $(\tilde{U}_{t})_{t
\in \theta}: \mathfrak{F}  \rightarrow
\mathcal{S}(H^{n|\mathbf{Y}|})$ as above. If the deterministic
capacity of $(\tilde{U}_{t})_{t \in \theta}$ is positive, we can
build, as above, a $(\nu(n)+n, J_n)$ common randomness-assisted
 code $\mathcal{C}(X,Y) =
\bigg\{\Bigl(E_{\mathbf{x}^{\nu(n)+n}},\{D_j^{\mathbf{y}^{\nu(n)+n}}:j\in \{ 1,\ldots
,J_n\}\}\Bigr):
\mathbf{x}^{\nu(n)+n}\in\mathbf{X}^{\nu(n)+n},\mathbf{y}^{\nu(n)+n}\in\mathbf{Y}^{\nu(n)+n}\bigg\}$
such that
\begin{align*} &\max_{t^{\nu(n)+n} \in \theta^{\nu(n)+n}} \sum_{\mathbf{x}^{\nu(n)+n}\in\mathbf{X}^{\nu(n)+n},
 \mathbf{y}^{\nu(n)+n}\in\mathbf{Y}^{\nu(n)+n}}p(\mathbf{x}^{\nu(n)+n},
 \mathbf{y}^{\nu(n)+n})\\
 & P_{e} (C(\mathbf{x}^{\nu(n)+n},\mathbf{y}^{\nu(n)+n}), t^{\nu(n)+n}) < \epsilon\text{ ,}\end{align*}
\begin{align*} &\max_{t^{\nu(n)+n}\in\theta^{\nu(n)+n}}  \sum_{\mathbf{y}^{\nu(n)+n}\in\mathbf{Y}^{\nu(n)+n}}
p\left(\mathbf{x}^{\nu(n)+n},\mathbf{y}^{\nu(n)+n}\right)\chi\left(R_{\mathrm{uni}};Z_{t^{\nu(n)+n},\mathbf{x}^{\nu(n)+n}}\right)
< \zeta\text{ .} \end{align*} But this would mean
\[C_s(\{(W_t,{V}_t): t \in \theta\};corr(X,Y))= C_s(\{(W_t,{V}_t): t \in \theta\};r) \text{ ,}\]
and there is nothing to prove.\vspace{0.15cm}

Thus we may assume that
the deterministic
capacity of $(\tilde{U}_{t})_{t \in \theta}$ is
 equal to zero. This implies that
 $(\tilde{U}_{t})_{t \in \theta}$ is symmetrizable
 (cf. \cite{Ahl/Bli}), i.e., there is
a parametrized set of distributions $\{\tau(\cdot\mid f):
 f\in \mathfrak{F}\}$ on $\theta$ such that for all
 $f$, $f' \in \mathfrak{F}$ we have

 \begin{align} &\sum_{t\in\theta}\tau(t\mid f')
\sum_{\mathbf{x}}\sum_{\mathbf{y}}
 P\left(\mathbf{X}\times\mathbf{Y}\right) { \breve{\kappa}}_{y} \otimes W_t\left(f(\mathbf{x})\right)
 =\sum_{t\in\theta} \tau(t\mid f)
\sum_{\mathbf{x}}\sum_{\mathbf{y}}
 P\left(\mathbf{X}\times\mathbf{Y}\right) { \breve{\kappa}}_{y} \otimes W_t\left(f'(\mathbf{x})\right)\notag\\
&\Rightarrow  \sum_{t\in\theta}\tau(t\mid f')
\sum_{\mathbf{x}}
 P\left(\mathbf{X}\times\mathbf{Y}\right)  W_t\left(f(\mathbf{x})\right)
 =\sum_{t\in\theta} \tau(t\mid f)
\sum_{\mathbf{x}}
 P\left(\mathbf{X}\times\mathbf{Y}\right)  W_t\left(f'(\mathbf{x})\right)\label{assumesymm}\end{align}
for all $\mathbf{y}\in\mathbf{Y}$.\vspace{0.15cm}

Our approach  is similar to the technique of  \cite{Ahl/Cai}. Let
$\mathbf{A}=\{0,1,\ldots,|\mathbf{A}|-1\}$,
 $\mathbf{X}=\mathbf{Y}=\{0,1\}$.
We define functions $g^*$ and $g_i\in \mathfrak{F}$ for $i=1,\ldots, a-1$
such that
$g^*(0)=g^*(1)=0$ and $g_i(u):=i+u$ mod $|\mathbf{A}|$ for $u\in \{0,1\}$.
Since $(\tilde{U}_{t})_{t \in \theta}$ is symmetrizable, by (\ref{assumesymm})
there is
a parametrized set of distributions $\{\tau(t\mid f):
 f\in \mathfrak{F}\}$ on $\theta$ such that for all $a\in \mathbf{A}$, the following two equalities are
valid
\begin{align*} &\sum_{t\in\theta}p(0,0)\tau(t\mid g^*) W_t(a)\\
&+ \sum_{t\in\theta}p(1,0)\tau(t\mid g^*) W_t(a+1\text{ mod }|\mathbf{A}|)\\
&=\sum_{t\in\theta}p(0,0)\tau(t\mid g_i) W_t(a)\\
&+ \sum_{t\in\theta}p(1,0)\tau(t\mid g_i) W_t(a)\\
&=\sum_{t\in\theta} \tau(t\mid g_i)W_t(a) \text{ ;}\end{align*}
\begin{align*} &\sum_{t\in\theta}p(0,1)\tau(t\mid g^*) W_t(a)\\
&+ \sum_{t\in\theta}p(1,1)\tau(t\mid g^*) W_t(a+1\text{ mod }|\mathbf{A}|)\\
&=\sum_{t\in\theta}p(0,1)\tau(t\mid g_i) W_t(a)\\
&+ \sum_{t\in\theta}p(1,1)\tau(t\mid g_i) W_t(a)\\
&=\sum_{t\in\theta}\tau(t\mid g_i) W_t(a) \text{ .}\end{align*}

 If we choose an arbitrary orthonormal basis on $H$ to write the following quantum states
 in form of matrices
 \[\left(m_{k,l}\right)_{k,l=1,\ldots,\dim H}= \sum_{t\in\theta}\tau(t\mid g^*) W_t(a)\text{ ,}\]
  \[\left({m'}_{k,l}\right)_{k,l=1,\ldots,\dim H}= \sum_{t\in\theta}\tau(t\mid g^*) W_t(a+1\text{ mod }|\mathbf{A}|)\text{ ,}\]
  \[\left({m*}_{k,l}\right)_{k,l=1,\ldots,\dim H}= \sum_{t\in\theta}\tau(t\mid g_i) W_t(a)\text{ ,}\]
 for all $k,l \in \{1,\ldots,\dim H\}$ we have
  \[p(0,0)m_{k,l}+p(1,0){m'}_{k,l}={m*}_{k,l}\text{ ,}\]
 \[p(0,1)m_{k,l}+p(1,1){m'}_{k,l}={m*}_{k,l}\text{ .}\]

  Since $I(X,Y)$ is positive, $p(0,0)\not=p(1,0)$ and $p(0,1)\not=p(1,1)$,
  therefore $\det\left(
 \begin{matrix} p(0,0)& p(1,0)\\
  p(0,1)& p(1,1)
 \end{matrix}\right)\not= 0$. Thus $m_{k,l}={m'}_{k,l}={m*}_{k,l}$
 for all $k,l \in \{1,\ldots,\dim H\}$, this means
 \[\sum_{t\in\theta}\tau(t\mid g^*) W_t(a)=\sum_{t\in\theta}\tau(t\mid g^*) W_t(a+1\text{ mod }|\mathbf{A}|)\]
 for all $a\in \mathbf{A}$.

 Therefore,
 for any $n\in\mathbb{N}$ and any given $(n, J_n)$  code $\mathcal{C}^{\gamma}$ $=$
 $\Bigl(E^{\gamma}, \{D_j^{\gamma}:j=1,\ldots,J_n\}\Bigr)$, the following statement is valid.  Let ${a'}^n$ be an arbitrary
 sequence in ${\mathbf{A}}^n$, we have
\begin{align}&\sum_{t^n\in\theta^n}\tau(t^n\mid g^*)P_e(\mathcal{C}^{\gamma},t^n)\notag\\
&=\sum_{t^n\in\theta^n}\tau(t^n\mid g^*)
 \left[1- \frac{1}{J_n} \sum_{j=1}^{J_n} \mathrm{tr}(W_{t^n}(E_{\gamma}(~|j))D_j^{\gamma})\right]\notag\\
&=\sum_{t^n\in\theta^n}\tau(t^n\mid g^*)
 \left[1- \frac{1}{J_n} \sum_{j=1}^{J_n}  \sum_{a^n \in {\mathbf{A}}^n}
E^{\gamma}(a^n|j)\mathrm{tr}(W_{t^n}(a^n)D_j^{\gamma})\right]\notag\\
&=
1- \frac{1}{J_n} \sum_{j=1}^{J_n}  \sum_{a^n \in {\mathbf{A}}^n}
E^{\gamma}(a^n|j)\mathrm{tr}\left(\sum_{t^n\in\theta^n}\tau(t^n\mid g^*)W_{t^n}(a^n)D_j^{\gamma}\right)\notag\\
&=
1- \frac{1}{J_n} \sum_{j=1}^{J_n}\sum_{a^n \in {\mathbf{A}}^n}
E^{\gamma}(a^n|j)  \mathrm{tr}\left(\sum_{t^n\in\theta^n}\tau(t^n\mid g^*)W_{t^n}({a'}^n)D_j^{\gamma}\right)\notag\\
&=
1- \frac{1}{J_n} \sum_{j=1}^{J_n}  \mathrm{tr}\left(\sum_{t^n\in\theta^n}\tau(t^n\mid g^*)W_{t^n}({a'}^n)D_j^{\gamma}\right)\notag\\
&=
1- \frac{1}{J_n}\sum_{t^n\in\theta^n}\tau(t^n\mid g^*)\mathrm{tr}\left(W_{t^n}({a'}^n) \sum_{j=1}^{J_n}  D_j^{\gamma}\right)\notag\\
&=
1- \frac{1}{J_n}\sum_{t^n\in\theta^n}\tau(t^n\mid g^*)\mathrm{tr}\left(W_{t^n}({a'}^n) \right)\notag\\
&= 1- \frac{1}{J_n}\text{ ,}\end{align}
 where $\tau(t^n\mid g^*) := \tau(t_1\mid g^*) \tau(t_2\mid g^*) \ldots \tau(t_n\mid g^*)$ for
 $t^n=(t_1,t_2,\ldots, t_n)$. The second and the fifth equations hold because the trace function
and matrices' multiplication are linear. The first, the fourth, and the last  equations hold because
$\sum_{t^n\in\theta^n}\tau(t^n\mid g^*)$ $=$ $\sum_{a^n \in {\mathbf{A}}^n}
E^{\gamma}(a^n|j)$ $=$ $\mathrm{tr}\left(W_{t^n}({a'}^n) \right)$ $=1$ for all $g^*$, $j$, and ${a'}^n$. The sixth
equation holds because $\sum_{j=1}^{J_n}  D_j^{\gamma}=id$.

Thus for  any $n\in\mathbb{N}$,   any
$J_n\in\mathbb{N}\setminus\{1\}$, and any $(n, J_n)$  randomness-assisted
quantum code $(\{\mathcal{C}^{\gamma}:\gamma\in \Lambda\},G)$  we have
\begin{align}&\frac{J_n-1}{J_n}\notag\\
&= \int_{\Lambda} \sum_{t^n\in\theta^n}\tau(t^n\mid g^*) P_{e}(\mathcal{C}^{\gamma},t^n)d
G(\gamma) \notag\\
& = \sum_{t^n\in\theta^n}\tau(t^n\mid g^*)\int_{\Lambda}  P_{e}(\mathcal{C}^{\gamma},t^n)d
G(\gamma) \notag\\
&=\mathbb{E}\left(\int_{\Lambda}  P_{e}(\mathcal{C}^{\gamma},\mathfrak{T}^n)d G(\gamma) \right)\label{swi}\text{ ,}\end{align}
  where $\mathfrak{T}^n$ is a random variable on $\theta^n$ such that
  $Pr(\mathfrak{T}^n=t^n)=\tau(t^n\mid g^*)$ for all $t^n\in\theta^n$.
 \vspace{0.15cm}

 By  (\ref{swi})  for  any $n\in\mathbb{N}$,   any
$J_n\in\mathbb{N}\setminus\{1\}$ and any $(n, J_n)$  random-assisted
 quantum code $(\{\mathcal{C}^{\gamma}:\gamma\in \Lambda\},G)$,
 there exists at least one  $t^n\in\theta^n$ such that
  \begin{equation}\int_{\Lambda}  P_{e}(\mathcal{C}^{\gamma},t^n)d G(\gamma)
  \geq \frac{J_n-1}{J_n}\text{ .}\label{swi2}\end{equation}\vspace{0.15cm}

By  (\ref{swi2}) for  any $n\in\mathbb{N}$,  any $J_n>1$, there is no $(n, J_n)$
randomness-assisted
 code $(\{\mathcal{C}^{\gamma}:\gamma\in \Lambda\},G)$ for
$\{(W_t,{V}_t): t \in \theta\}$ such that
\[ \max_{t^n\in\theta^n}\int_{\Lambda}P_{e}(\mathcal{C}^{\gamma},t^n)d
G(\gamma) < \frac{1}{2}\text{ ,}\]
  therefore if the
  $m-a-(X,Y)$
      secrecy capacity of $\{(W_t,{V}_t): t \in \theta\}$ is
 equal to zero and $I(X,Y)$ is positive,
  the randomness-assisted secrecy capacity  of $\{(W_t,{V}_t): t \in \theta\}$ is
 equal to $\log 1 =0$. But this is a contradiction to our assumption that
$C_s(\{(W_t,{V}_t): t \in \theta\};r)$ is positive.\vspace{0.2cm}

 This result and the result for the case
 when the
  $m-a-(X,Y)$
      secrecy capacity of $\{(W_t,{V}_t): t \in \theta\}$ is
  positive complete our proof for Theorem \ref{causalcommonrandom}.
  \end{proof}

Theorem \ref{causalcommonrandom} shows that  the correlation is
 a very helpful resource
 for the secure message transmission through an arbitrarily varying classical-quantum
 wiretap channel. As Example \ref{nachhinten} shows, there are indeed arbitrarily varying classical-quantum
 wiretap channels which have zero deterministic secrecy capacity, but 
at the same time positive random  secrecy capacity.
Theorem \ref{causalcommonrandom} shows that if we   have a
$m-a-(X,Y)$ correlation as a resource, even when it is insecure and 
very weak (i.e. $I(X,Y)$  needs only to be  slightly
larger than zero), these channels will have a positive  $m-a-(X,Y)$
secrecy capacity.

\section{Applications and Further Notes }\label{AP}

In Subsection \ref{FNoRT}  we will
discuss the importance of the Ahlswede dichotomy for arbitrarily
varying classical-quantum
wiretap  channels.
We will show that it
can  occur  that the deterministic capacity of
an arbitrarily varying classical-quantum
 wiretap channel is not equal to its randomness-assisted capacity.

%In  Subsection \ref{TDMfEG} we will give an example which
%shows that quantum channel
%with channel uncertainty  and eavesdropping is indeed  one of the most
%difficult areas. As we already discussed in  Section \ref{BDaCS},
%we still cannot  find a 
%result similar to the  Ahlswede dichotomy for the entanglement generating capacity of an arbitrarily
%varying   quantum
% channel
%(cf. Conjecture \ref{AhlswedeConjecture}). The example in
% Subsection \ref{TDMfEG} shows that some techniques  for quantum channels
%without channel uncertainty do not work
%if we consider quantum channel
%with channel uncertainty  and eavesdropping. 

In Subsection \ref{SuAc} we will show that the research in quantum
channels not only sets limitations, but also offers new fascinating
possibilities. Applying the  Ahlswede dichotomy, we can prove 
that  two arbitrarily varying classical-quantum
 wiretap channels, both with zero security capacity,
 allow perfect secure transmission, if we use them together. This is a phenomenon
called ``super-activation'' which appears in 
quantum information theory (cf. \cite{Li/Wi/Zou/Guo}).

\subsection{Further Notes on Resource Theory}\label{FNoRT}

In this subsection, we give some notes on resource theory and the Ahlswede dichotomy.\vspace{0.2cm}

%1) While in  classical information theory single-letter formulas
%are available for  most capacities results, one of the main
% discoveries in   quantum information theory is that
%for the quantum channels, most of the capacities results can only be
%expressed in multi-letter formulas (cf. \cite{DiVi/Sh/Sm} and
%\cite{Ha}).

%Unlike   most  capacities results for  quantum channels, however,
%we have the surprising fact that
% the condition for the zero deterministic  security capacity of
%an arbitrarily varying classical-quantum
% wiretap channel can be expressed in a single-letter formula
%(cf. Definition \ref{symmet}). \vspace{0.2cm}

1) The Ahlswede dichotomy states that either the deterministic
security capacity of an arbitrarily varying classical-quantum
 wiretap channel is zero or it equals its
randomness-assisted
 security  capacity.
There are actually arbitrarily varying classical-quantum
 wiretap channels which have zero deterministic  security capacity,
 but achieve a positive security capacity if the sender and the
 legal receiver can use a resource, as the following example shows.
 This shows that the Ahlswede dichotomy
 is indeed a ``dichotomy'', and how helpful  a resource  can be for
 the robust and secure message transmission.

\begin{example}\label{nachhinten}Let  $\{(W_t,{V}_t): t \in \theta\}$ be an arbitrarily varying classical-quantum
 wiretap channel. By Theorem \ref{dichpart}. \ref{dichpart1}, $C_s(\{(W_t,{V}_t): t \in \theta\})$ is equal to
$C_s(\{(W_t,{V}_t): t \in \theta\};r)$ if $\{W_t : t \in \theta\}$ is not
symmetrizable, and equal to zero if $\{W_t : t \in \theta\}$ is 
symmetrizable.
 If  $\{W_t : t \in \theta\}$ is
symmetrizable,
 it can actually occur that
 $C_s(\{(W_t,{V}_t): t \in \theta\})$ is zero, but
 $C_s(\{(W_t,{V}_t): t \in \theta\};r)$ is positive, as  following   example shows
(c.f. \cite{Ahl/Bli} for the case of an arbitrarily varying classical-quantum channel  without
wiretap).

Let $\theta:=\{1,2\}$. Let $\mathbf{A}$ $=$ $\{0,1\}$.
Let ${H}^{\mathfrak{B}}$ $=$ $\mathbb{C}^{3}$.
Let $\{|0\rangle^{\mathfrak{B}}, |1\rangle^{\mathfrak{B}}, 
|2\rangle^{\mathfrak{B}}$ be a set of orthonormal vectors
on ${H}^{\mathfrak{B}}$.

For $r\in[0,1]$ let $P_{r}$ be the probability distribution on $\mathbf{A}$
such that $P_{r}(0)=r$ and $P_{r}(1)=1-r$.
 We define a  channel $W_{1}$ $:P(\mathbf{A})$ $\rightarrow$ $\mathcal{S}({H}^{\mathfrak{B}})$
by 
\[W_{1} (P_{r})=  r |0\rangle\langle 0|^{\mathfrak{B}} 
+ (1-r)|1\rangle\langle 1|^{\mathfrak{B}}\text{ ,}\]
and a  channel $W_{2}$ $:P(\mathbf{A})$ $\rightarrow$ $\mathcal{S}({H}^{\mathfrak{B}})$
by 
\[W_{1} (P_{r})=  r |1\rangle\langle 1|^{\mathfrak{B}} 
+  (1-r)|2\rangle\langle 2|^{\mathfrak{B}}\text{ .}\]

In other word
\[W_{1} (0)=  |0\rangle\langle 0|^{\mathfrak{B}} \text{ ,}\]
\[W_{1} (1)=   |1\rangle\langle 1|^{\mathfrak{B}}\text{ ,}\]

\[W_{2} (0)=  |1\rangle\langle 1|^{\mathfrak{B}} \text{ ,}\]
\[W_{2} (1)=   |2\rangle\langle 2|^{\mathfrak{B}}\text{ .}\]\vspace{0.2cm}

Let ${H}^{\mathfrak{E}}$ $=$ $\mathbb{C}^{2}$.
Let $\{|3\rangle^{\mathfrak{E}}, |4\rangle^{\mathfrak{E}}$ be a set of orthonormal vectors
on ${H}^{\mathfrak{E}}$.

 We define a  channel $V_{1}$ $:P(\mathbf{A})$ $\rightarrow$ $\mathcal{S}({H}^{\mathfrak{E}})$
by 
\[V_{1} (P_{r})= |3\rangle\langle 3|^{\mathfrak{E}}\text{ ,}\]
and a  channel $V_{2}$ $:P(\mathbf{A})$ $\rightarrow$ $\mathcal{S}({H}^{\mathfrak{E}})$
by 
\[V_{1} (P_{r})= |4\rangle\langle 4|^{\mathfrak{E}}\text{ .}\]\vspace{0.2cm}

$\{(W_t,{V}_t): t \in \theta\}$
defines an arbitrarily
varying classical-quantum wiretap channel.\vspace{0.2cm}

We set
\begin{align*}&\tau(1\mid 0) = 0\text{ ; } ~~\tau(2\mid 0) = 1\text{ ;}\\
&\tau(1\mid 1) = 1\text{ ; } ~~\tau(2\mid 1) = 0\text{ .}\end{align*}

It holds
\[\sum_{t\in\theta}\tau(t\mid 0)W_{t}({0})=\sum_{t\in\theta}\tau(t\mid {0})W_{t}(0)\text{ ,}~
\sum_{t\in\theta}\tau(t\mid 1)W_{t}({1})=\sum_{t\in\theta}\tau(t\mid {1})W_{t}(1)\text{ ,}\]
and
\[\sum_{t\in\theta}\tau(t\mid 0)W_{t}({1})=|1\rangle\langle 1|^{\mathfrak{E}} 
=\sum_{t\in\theta}\tau(t\mid {1})W_{t}(0)\text{ .}\]
$\{(W_t): t \in \theta\}$ is therefore symmetrizable.
 By Theorem \ref{dichpart}. \ref{dichpart1}, we have \begin{equation}C_s(\{(W_t,{V}_t): t \in \theta\})=0\text{ .}\end{equation} \vspace{0.15cm}

By \cite{Bl/Ca}, for any arbitrarily varying classical-quantum wiretap channel $\{(W_t,{V}_t): t \in \theta\}$, we have
\begin{align}&C_s(\{(W_t,{V}_t): t \in \theta\};r)\notag\\
& \geq \max_{P\in\mathcal{P}}
\biggl(\min_{Q\in\mathcal{Q}}\chi\left(P,\{U^Q(a): a\in {\mathbf{A}}\}\right) \notag\\
&-\lim_{n\rightarrow\infty}\max_{t^n\in\theta^n}\frac{1}{n}\chi(P^n,\{{V}_{t^n}(a^n): a^n\in {\mathbf{A}}^n\})\biggr)\text{
,}\label{t1}\end{align} where $\mathcal{P}$ is the set of distributions on
$\mathbf{A}$, $\mathcal{Q}$ is the set of distributions on $\theta$, and
$U^Q(a) = \sum_{t\in\Theta}Q(t)W_{t}(a)$ for
$Q\in\mathcal{Q}$.

For  all $n\in\mathbb{N}$, $t^n\in\theta^n$, and $P^n\in\mathcal{P}^n$, we have
 $\chi(P^n,\{{V}_{t^n}(a^n): a^n\in {\mathbf{A}}^n\}) = 1 \log 1 - 1 \log 1 = 0$ and therefore
\[C_s(\{(W_t,{V}_t): t \in \theta\};r) \geq \max_{P\in\mathcal{P}}
\min_{Q\in\mathcal{Q}}\chi\left(P,\{U^Q(a): a\in {\mathbf{A}}\}\right)\text{ .}\]

We denote by $p'\in P(A)$ the distribution on $A$ such that $p'(1)=p'(2)=\frac{1}{2}$.
Let $q\in[0,1]$.  We define $Q(1)=q$,   $Q(2)=1-q$. We have
\begin{align*}
&\chi\left(p',\{W_Q^{0}(a): a\in {\mathbf{A}}\}\right)\\
&=-\frac{1}{2}q\log \frac{1}{2}q + \frac{1}{2}(1-q)\log\frac{1}{2}(1-q)
-\frac{1}{2}\log\frac{1}{2}\\
&+q\log q + (1-q)\log(1-q)\text{ .}\end{align*}
By the differentiation  by $q$,
we obtain
\begin{align*}&\frac{1}{\log e}\biggl(-\frac{1}{2}\log \frac{1}{2}q - \frac{1}{2}+ 
\frac{1}{2}\log\frac{1}{2}(1-q)+\frac{1}{2}
+\log q +1- \log(1-q)-1\biggr)\\
&=\frac{1}{2\log e}\left(\log q - \log(1-q)\right)\text{ .}\end{align*}
This term is equal to zero if and only if $q=\frac{1}{2}$. By further calculation,
one can show
that  $\chi\left(p',\{W_Q^{0}(a): a\in {\mathbf{A}}\}\right)$ achieves its minimum when  $q=\frac{1}{2}$.  This
minimum is equal to $-\frac{1}{2}\log\frac{1}{4}+\frac{1}{2}\log\frac{1}{2}$ $=$ $\frac{1}{2}$ $>0$.
Thus \[\max_p\min_{q} \chi\left(p,B_q^{0}\right)\geq\frac{1}{2}\text{ .}\]

For  all $t\in \theta$, it holds
${V}_{t}^{0} (0) = {V}_{t}^{0} (1)$ and
therefore for  all $t^n\in \theta^n$ and any $p^n\in P(A^n)$,  we have
\begin{align*}&\chi(p;Z_{t^n}^{0})\\
&= S({V}_{t^n}^{0} (p^n) ) - \sum_{a^n\in A^n} p^n(a^n)S({V}_{t^n}^{0} (a^n) ) \\
&=0\text{ .}\end{align*} By (\ref{t1}), 

\begin{equation} C_s(\{(W_t^{0},{V}_t^{0}): t \in \theta\},cr) \geq  \frac{1}{2}-0 >0\text{ .}\label{foracswtvtcr1}\end{equation}
 \vspace{0.2cm}

This shows an example of
 an arbitrarily varying classical-quantum channel such that
 its deterministic  capacity  is zero, but its random  capacity  is positive.
 \label{nachhinten1}
 \end{example}

 Thus, a ``useless'' arbitrarily varying classical-quantum channel,
 i.e., with zero deterministic  secrecy capacity,  allows
 secure transmission if the sender and the legal receiver have the possibility
 to use a resource, either randomness, common randomness, or even a
 ``cheap'', insecure, and weak correlation. Here we say ``cheap'' and
 ``weak'' in the sense of the discussion in Section
\ref{AVCQWCWCA}.

\subsubsection{Super-Activation}\label{SuAc}

One of the properties of classical channels is that in the majority
of cases, if we have a channel system where two sub-channels are
used together, the capacity of this channel system is the sum of the
 two sub-channels' capacities. Particularly,
 a system consisting of two orthogonal classical
channels, where both are ``useless''  in the sense  that they both have zero
capacity for message transmission, the capacity  for message transmission
of the whole system is zero as well (``$0 + 0 = 0$''). 
 For the definition of   ``two orthogonal
 channels'' in classical systems, please see  \cite{Fa/Ka}.

In contrast to the classical information  theory, it is known that
the capacities of quantum  channels can be super-additive, i.e.,
there are cases in which the capacity of the product $W_1\otimes W_2$
%$:A_1 \times A_2 \rightarrow \mathcal{S}(H_1 \otimes H_2)$, $(a,b) \rightarrow
%W_1(a)\otimes W_2(b)$,
 of
 two  quantum  channels $W_1$
 %$:A_1 \rightarrow \mathcal{S}(H_1)$,
%$a \rightarrow W_1(a)$,
 and $W_2$
 %$:A_2 \rightarrow \mathcal{S}(H_2)$, $b \rightarrow W_2(b)$,
is
larger than the sum of the capacity of $W_1$ and
 the capacity of $W_2$ (cf. \cite{Li/Wi/Zou/Guo} and \cite{Gi/Wo}).
 %Here $A_1$ and $A_2$ are finite sets and $H_1$ and $H_2$ are Hilbert
%spaces
``The whole is greater than the sum of its parts'' - Aristotle.

 Particularly in quantum information  theory, there are examples of two quantum
channels, $W_1$ and $W_2$,  with zero capacity, which
 allow perfect transmission if they are used together, i.e.,  the   capacity of
their product $W_1\otimes W_2$ is positive, (cf. \cite{Smi/Yar},
\cite{Smi/Smo/Yar}, \cite{Opp} and also \cite{Bo/Wy} for a rare case
result when this phenomenon occurs using two classical arbitrarily
varying wiretap channels). This is due to the fact that there are
different reasons why a quantum channel can have zero capacity. We
call this phenomenon   ``super-activation'' (``$0+0 >0$'').

It is known that
 arbitrarily varying classical-quantum  wiretap
channels with positive secrecy capacities  are super-additive.
This means that the product $W_1\otimes W_2$ of
 two  arbitrarily varying classical-quantum  wiretap  channels $W_1$ and $W_2$, both
with positive secrecy capacities,  can have a capacity which
is larger than the sum of the
capacity of $W_1$ and
 the capacity of $W_2$ (cf. \cite{Li/Wi/Zou/Guo}).\vspace{0.15cm}

Using Theorem \ref{dichpart}. \ref{dichpart}, we can
 demonstrate the following Theorem,
\begin{theorem}
Super-activation occurs for  arbitrarily varying classical-quantum  wiretap
channels.\label{superactivation}\end{theorem}

Please note that the results of \cite{Li/Wi/Zou/Guo} (super-additivity of  
arbitrarily varying classical-quantum  wiretap
channels with positive secrecy capacities) do not imply super-activation of  
arbitrarily varying classical-quantum  wiretap
channels, since here we consider channels with zero secrecy
capacity.

 We will prove Theorem  \ref{superactivation} by giving an example (Example \ref{examplesupact}) in which two
arbitrarily varying classical-quantum  wiretap
channels,   which are themselves  ``useless''  in the sense  that they have both zero
 secrecy capacity, acquire positive  secrecy capacity when
used together. This is due the following.

Suppose we have an
 arbitrarily varying classical-quantum  wiretap  channel with positive
randomness-assisted secrecy
capacity.
By Theorem \ref{dichpart}. \ref{dichpart2},  the randomness-assisted secrecy
capacity is equal to the common randomness-assisted secrecy
capacity.
 But the problem for the sender and the legal receiver is that each
party does not know  which code is used in the
particular transmission if the channel that connects them
 has zero deterministic  capacity for message
transmission. However, suppose we have  another arbitrarily varying
classical-quantum
 wiretap  channel which has a positive  deterministic  capacity for message transmission.
Then the sender and the legal receiver can use it to transmit
 which code is used in the particular
transmission. This is possible even when the second arbitrarily
varying  classical-quantum
 wiretap  channel has
 zero randomness-assisted secrecy
capacity, since we allow the  wiretapper
to know  which specific code is 
 used.

We may see it in the following way. If we have two arbitrarily
varying  classical-quantum
 wiretap  channels, one of them is relatively  secure, but not very robust against
 jamming, while the other one is relatively robust, but not very secure against
 eavesdropping. We can
achieve that they ``remove'' their weaknesses from each other, or, in
other words, ``activate'' each other.\vspace{0.2cm}

We now give an example of super-activation for arbitrarily varying classical-quantum  wiretap  channels.\vspace{0.15cm}

\begin{example}
Let $\theta=\{1,2\}$, $\mathbf{A}=\{0,1\}$, and let $H=H'$ be
spanned by the orthonormal vectors $|0\rangle$ and  $|1\rangle$. We define
$\{(W_t,{V}_t): t \in \theta\}$ as in Example \ref{nachhinten1}.
We define
$\{({W'}_t,{V'}_t):t \in \theta\}$ by
\begin{align}&{W'}_1(0) = \frac{3}{4}|0\rangle\langle 0|+\frac{1}{4}|1\rangle\langle 1|\notag\\
&{W'}_1(1) = \frac{1}{4}|0\rangle\langle 0|+\frac{3}{4}|1\rangle\langle 1|\notag\\
&{W'}_2(0) = \frac{3}{4}|0\rangle\langle 0|+\frac{1}{4}|1\rangle\langle 1|\notag\\
&{W'}_2(1) =\frac{1}{4}|0\rangle\langle 0|+\frac{3}{4}|1\rangle\langle 1|\notag\\[0.15cm]
&{V'}_1(0) = |0\rangle\langle 0|\notag\\
&{V'}_1(1) = |1\rangle\langle 1|\notag\\
&{V'}_2(0) = |0\rangle\langle 0|\notag\\
&{V'}_2(1) = |1\rangle\langle 1|
\text{ .}\label{example2}\end{align}

We denote   the uniform distribution on $\mathbf{A}$ by $P$.
We have $P(0)=P(1)=\frac{1}{2}$. By \cite{Ahl/Bli}
the capacity  of $\{{W'}_t: t \in \theta\}$ is larger or equal to $\min_{Q\in\mathcal{Q}}\chi\left(P,\{U^Q(a): a\in {\mathbf{A}}\}\right)$
$=\frac{1}{2}-\frac{3}{4}\log\frac{3}{4}>0$.\vspace{0.15cm}

However, for all $(n,J_n)$ code $\bigl(E^n, \{D_j^n : j = 1,\ldots, J_n\}\bigr)$
the wiretapper can define a set of decoding operators
$\{D_{j,wiretap}^n: j= 1,\ldots J_n\}$ by $D_{j,wiretap}^n:=\sum_{a^n}E^n\left(a^n\mid j\right)
\left(\bigotimes_{i}|a_i\rangle\right)\left(\bigotimes_{i}\langle a_i|\right)$.
For any   probability distribution   $Q^n$
on ${\mathbf{A}}^n$, denote the wiretapper's random output using
$\{D_{j,wiretap}^n: j= 1,\ldots, J_n\}$ at channel state $t^n$ by $C_{t^n}$,
then $\chi(Q^n,Z_{t^n})\geq I(Q^n,C_{t^n})=H(Q^n)$, where $I(\cdot,\cdot)$ is
the mutual information, and $H(\cdot)$ is the Shannon entropy
(please cf. \cite{Cs/Ko} for the definitions of the mutual information and  the Shannon entropy
for classical random variables). If $\chi(R_{\mathrm{uni}},Z_{t^n})< \frac{1}{2}$ holds,
we also have $\log J_n = H(R_{\mathrm{uni}})< \frac{1}{2}$, but this implies $J_n=1$. Thus
\begin{equation}C_s(\{({W'}_t,{V'}_t):t \in \theta\})=0\text{ .}\end{equation}\vspace{0.2cm}

Let us now consider the arbitrarily
varying  classical-quantum wiretap channel $\biggl\{\Bigl(W_{t_1}\otimes {W'}_{t_2},{V}_{t_1}\otimes  {V'}_{t_2}\Bigr):(t_1,t_2) \in \theta^2\biggr\}$,
where  $\Bigl\{\Bigl(W_{t_1}\otimes {W'}_{t_2}\Bigl):(t_1,t_2) \in \theta^2\Bigr\}$ is an arbitrarily
varying  classical-quantum channel
$\{(00),(01),(10),(11)\}\rightarrow H^{\otimes 2}$, $(a,a')\rightarrow W_{t_1} (a) \otimes  {W'}_{t_2} (a')$,
and  $\Bigl\{\Bigl({V}_{t_1}\otimes  {V'}_{t_2}\Bigl):(t_1,t_2) \in \theta^2\Bigr\}$ is an arbitrarily
varying  classical-quantum channel
$\{(00),(01),(10),(11)\}\rightarrow H^{\otimes 2}$, $(a,a')\rightarrow {V}_{t_1} (a) \otimes  {V'}_{t_2} (a')$,
if the channel state is $(t_1,t_2)$.\vspace{0.15cm}

We have
\begin{equation}C_s\left(\biggl\{\Bigl(W_{t_1}\otimes {W'}_{t_2},{V}_{t_1}\otimes  {V'}_{t_2}\Bigr):(t_1,t_2) \in \theta^2\biggr\};r\right) \geq
\frac{1}{2}>0
\text{ .}\end{equation}\vspace{0.15cm}

Assume $\Bigl\{\Bigl(W_{t_1}\otimes {W'}_{t_2}\Bigl):(t_1,t_2) \in \theta^2\Bigr\}$
is symmetrizable, then there exists
 a
parametrized set of distributions $\{\tau(\cdot\mid (a,a')):
 (a,a')\in  \{(00),(01),(10),(11)\}\}$ on $\theta^2$ such that for all $(a,a')$, $(b,b')\in \{(00),(01),(10),(11)\}$ it holds
\begin{align}&\sum_{(t_1,t_2) \in \theta^2}\tau((t_1,t_2)\mid (b,b'))W_{t_1} (a) \otimes  {W'}_{t_2} (a')\notag\\
&=\sum_{(t_1,t_2) \in \theta^2}\tau((t_1,t_2)\mid (a,a'))W_{t_1} (b) \otimes  {W'}_{t_2} (b')\text{ .}\label{sowieBoWy}\end{align}

(\ref{sowieBoWy}) implies that
\begin{align}&\sum_{(t_1,t_2) \in \theta^2}\tau((t_1,t_2)\mid (0,0))W_{t_1} (0) \otimes  {W'}_{t_2} (1)\notag\\
&=\sum_{(t_1,t_2) \in \theta^2}\tau((t_1,t_2)\mid (0,1))W_{t_1} (0) \otimes  {W'}_{t_2} (0)\notag\\
& \Rightarrow  \left(\tau((1,1)\mid (0,0))+\tau((1,2)\mid (0,0))\right)|0\rangle\langle 0|\otimes
 \left(\frac{1}{4}|0\rangle\langle 0|+\frac{3}{4}|1\rangle\langle 1|\right) \notag\\
&~ + \left(\tau((2,1)\mid (0,0))+\tau((2,2)\mid (0,0))\right)|1\rangle\langle 1|\otimes
 \left(\frac{1}{4}|0\rangle\langle 0|+\frac{3}{4}|1\rangle\langle 1|\right) \notag\\
&~  =   \left(\tau((1,1)\mid (0,1))+\tau((1,2)\mid (0,1))\right)|0\rangle\langle 0|\otimes
 \left(\frac{3}{4}|0\rangle\langle 0|+\frac{1}{4}|1\rangle\langle 1|\right) \notag\\
&~ + \left(\tau((2,1)\mid (0,1))+\tau((2,2)\mid (0,1))\right)|1\rangle\langle 1|\otimes
 \left(\frac{3}{4}|0\rangle\langle 0|+\frac{1}{4}|1\rangle\langle 1|\right) \notag\\
& \Rightarrow \left(\tau((1,1)\mid (0,0))+\tau((1,2)\mid (0,0))\right) =
9\left(\tau((1,1)\mid (0,0))+\tau((1,2)\mid (0,0))\right) \text{ and } \notag\\
&~  \left(\tau((2,1)\mid (0,0))+\tau((2,2)\mid (0,0))\right) =
9 \left(\tau((2,1)\mid (0,0))+\tau((2,2)\mid (0,0))\right) \notag\\
& \Rightarrow \lightning
\text{ .}\label{sowieBoWye}\end{align}
\vspace{0.15cm}

Therefore $\Bigl\{\Bigl(W_{t_1}\otimes {W'}_{t_2}\Bigl):(t_1,t_2) \in \theta^2\Bigr\}$
is not symmetrizable, and by Theorem \ref{dichpart}. \ref{dichpart1},
\begin{align}&C_s\left(\biggl\{\Bigl(W_{t_1}\otimes {W'}_{t_2},{V}_{t_1}\otimes  {V'}_{t_2}\Bigr):(t_1,t_2) \in \theta^2\biggr\}\right) \notag\\
&= C_s\left(\biggl\{\Bigl(W_{t_1}\otimes {W'}_{t_2},{V}_{t_1}\otimes  {V'}_{t_2}\Bigr):(t_1,t_2) \in \theta^2\biggr\};r\right) \notag\\
&>0\text{ .}\end{align} \label{examplesupact}\end{example}

This example shows that although both
$\{(W_t,{V}_t): t \in \theta\}$ and  $\{({W'}_t,{V'}_t):t \in \theta\}$ are themselves useless, they
 allow secure transmission using together (``$0+0 >0$''). Thus Theorem  \ref{superactivation}
is proven. This shows that the research in quantum channels with
channel uncertainty and eavesdropping can   lead to some promising
 applications.

\section{Conclusion}\label{concl}
In this paper, we studied message transmission over
 a classical-quantum  channel with both a jammer and an eavesdropper, which is
called an arbitrarily varying classical-quantum  wiretap  channel.
We also studied how helpful various resources  can be.

The Ahlswede dichotomy for classical arbitrarily varying channels
was introduced  in \cite{Ahl1}. The
Ahlswede dichotomy for arbitrarily varying classical-quantum
channels was established in \cite{Ahl/Bj/Bo/No}. In our paper, we
have generalized the result of \cite{Bl/Ca}  by  establishing the
Ahlswede dichotomy for arbitrarily varying classical-quantum wiretap
channels:
 Either the deterministic secrecy capacity of an
arbitrarily varying classical-quantum wiretap channel is zero, or it
equals its randomness-assisted secrecy
 capacity. Interestingly, the Ahlswede dichotomy shows that the
deterministic capacity for secure message transmission is, in general, 
not specified by entropy quantities. This is a new  behavior in communication
 due to active wiretap attacks.

Dealing with channel uncertainty and eavesdropping is one of the main tasks
 in modern communication systems, caused, for example, by hardware imperfection.
For  practical implementation, a
reasonable assistance for the transmitters
is to share  resources.
For example, in   wireless communication, the communication service
may send some  signals via satellite to its users.
Hence, we analyzed the secrecy capacities of various coding
schemes with resource assistance. 
 A surprising and
promising result of this paper is that  the resources do  not have to be secure
themselves to be helpful for secure message
transmission considering channel uncertainty. Another interesting 
fact is that in \cite{Bo/No}, it has
been shown that the correlation is a much ``cheaper'' resource than
  randomness and   common randomness. However, the results in
this paper show that for secure message transmission considering channel uncertainty.
   correlation is as  helpful  as
  randomness and   common randomness. Furthermore, a correlation
$(X,Y)$ does not have to be ``very good'' to be helpful in achieving a positive
secrecy capacity, since $(X,Y)$ is a helpful resource even if
$I(X,Y)$  is only slightly larger than
zero. We also gave an example that shows not only theoretically, but
also physically, how helpful a resource can be. In this example,
an arbitrarily varying classical-quantum wiretap channel has zero
deterministic secrecy capacity, but as soon as the sender and the
receiver can use a resource, either randomness, common randomness,
or correlation, we can achieve positive secrecy capacity. This
example  shows that for  communication in  practice, having   weak
public signals will be 
 very useful.

 %Moreover, we  gave an example which shows how
 %difficult it can be to find a  result similar to the
%Ahlswede dichotomy for the entanglement generating capacity of an
%arbitrarily varying quantum
% channel, and why it is so difficult. If we consider classical-quantum
%   channel with channel uncertainty and  eavesdropping, we show  that many new
%   kinds of
%   difficulties can appear.
% This shows that the Ahlswede dichotomy for arbitrarily varying
%classical-quantum wiretap  channels is  a  result  which is far from
%trivial.

In \cite{Smi/Yar} and \cite{Smi/Smo/Yar}, it has been shown that the
phenomenon  ``super-activation'' can occur for certain quantum
channels (``$0+0>0$''). In this paper, we have proved that
``super-activation'' occurs for  arbitrarily varying
classical-quantum  wiretap channels. In   classical information
theory, adding a telegraph wire that relays  no information  to a
system does not help in the majority of cases. Our result shows that
for message transmission over
 classical-quantum  channels with both a jammer and an eavesdropper,
adding a fiber-optic cable that relays  non-secure information can be
really useful. This result sets a new challenging task for the design of
media access control, which is an important topic for standardization
and certification. Unlike in   classical communication, for quantum
media access control, we have to consider that we can lose security if we
have two orthogonal useless arbitrarily varying classical-quantum wiretap channels.
To provide security, we therefore need a  more sophisticated design of
media access control than in the classical case. For example, 
we have to avoid
two useless channels to be orthogonal.

\section*{Acknowledgment}
Support by the Bundesministerium f\"ur Bildung und Forschung (BMBF)
via Grant 16KIS0118K and  16KIS0117K is gratefully acknowledged.


\begin{thebibliography}{xxx}
\bibitem{Ahl0} R. Ahlswede, A note on the existence of the weak capacity for channels with arbitrarily
varying channel probability functions and its relation to Shannon's zero error
capacity, The Annals of Mathematical Statistics, Vol. 41, No. 3, 1970.
\bibitem{Ahl-feedback} R. Ahlswede, Channels with Arbitrarily Varying Channel Probability 
Functions in the Presence of Noiseless Feedback, Z. Wahrscheinlichkeitstheorie verw. Gebiete, Vol. 25, 239-252, 
1973.
\bibitem{Ahl1} R. Ahlswede,  Elimination of correlation in random codes for arbitrarily varying channels,
Z. Wahrscheinlichkeitstheorie verw. Gebiete,
Vol. 44, 159-175, 1978.
\bibitem{Ahl2} R. Ahlswede,  Coloring hypergraphs: a new approach to multi-user
source coding-II,
Journal of Combinatorics, Information \& System
Sciences,
Vol. 5, No. 3, 220-268, 1980.
\bibitem{Ahl3} R. Ahlswede, Arbitrarily varying channels with states
sequence
known to the sender,
IEEE Trans. Inf. Th.,
Vol. 32, 621-629, 1986.
\bibitem{Ahl/Bj/Bo/No}R. Ahlswede,  I. Bjelakovi\'{c}, H. Boche,  and J. N\"otzel,
Quantum capacity under adversarial quantum noise:  arbitrarily
varying quantum channels, Comm. Math. Phys. A, Vol. 317, No. 1, 103-156,
2013.
\bibitem{Ahl/Bli} R. Ahlswede and  V. Blinovsky,  Classical capacity of classical-quantum arbitrarily
varying channels, IEEE Trans. Inform. Theory, Vol. 53, No. 2,
526-533, 2007.
\bibitem{Ahl/Cai} R. Ahlswede and N. Cai, Correlation sources help transmission
over an arbitrarily varying channel, IEEE
Trans. Inform. Theory, Vol. 43, No. 4,  1254-1255, 1997.
%\bibitem{Ahl/Csi1} R. Ahlswede and I.  Csisz\'ar,  Common randomness in information theory and cryptography, Part I: secret sharing. IEEE
%Trans. Inform. Theory, Vol. 39,   1121-1132,
%1993.
%\bibitem{Ahl/Csi2} R. Ahlswede and I.  Csisz\'ar,  Common randomness in information theory and cryptography, Part II: CR capacity. IEEE
%Trans. Inform. Theory, Vol. 44,   225-240, 1998.
%\bibitem{Ba/Ni/Sch} H. Barnum, M. A. Nielsen, and B. Schumacher, Information transmission through a noisy
%quantum channel, Phys. Rev. A, Vol. 57,  4153 , 1998.
\bibitem{Be} C. H. Bennet, Quantum cryptography using any two non-orthogonal states, Physical
Review Letters, Vol. 68, 3121-3124,  1992.
\bibitem{Be/Br} C. H. Bennett and G. Brassard, Quantum cryptography: public key distribution and coin tossing,
 Proceedings of the IEEE International Conference on Computers, Systems, and Signal Processing, Bangalore,  175, 1984.
%\bibitem{Be/Sh/Sm/Th} C. H. Bennett, P. W. Shor, J. A. Smolin, and A. V. Thapliyal,
% Entanglement-assisted capacity of a quantum channel and the reverse Shannon theorem.
%IEEE Trans. Inform. Theory, Vol.  48, 2637-2655, 2002.
\bibitem{Bj/Bo/Ja/No} I. Bjelakovi\'{c}, H. Boche, G. Jan\ss en, and J. N\"otzel,
Arbitrarily varying and compound classical-quantum channels and a
note on quantum zero-error capacities,  Information Theory, Combinatorics, and Search Theory: In Memory of Rudolf Ahlswede, Lecture Notes in Computer Science, H. Aydinian, F. Cicalese, and C. Deppe eds., LNCS
Vol. 7777,  247-283, Heidelberg: Springer Verlag, \tt
arXiv:1209.6325\rm, 2013.
%\bibitem{Bj/Bo/No2} I. Bjelakovi\'{c}, H. Boche,  and J. N\"otzel, Entanglement transmission capacity of compound channels,
% Proc. of International Symposium on Information Theory ISIT 2009,  Korea,  2009.
\bibitem{Bj/Bo/So} I. Bjelakovic, H. Boche, and J. Sommerfeld,
Secrecy  results for compound wiretap channels, Problems of Information Transmission,
Vol. 59, No. 3, 1405-1416, 2013.
\bibitem{Bj/Bo/So2}I. Bjelakovi\'{c}, H. Boche,
and J. Sommerfeld, Capacity results for arbitrarily varying wiretap
channels, 
Information Theory, Combinatorics, and Search Theory: In Memory of Rudolf Ahlswede, Lecture Notes in Computer Science, H. Aydinian, F. Cicalese, and C. Deppe eds., LNCS
Vol.7777, 123-144, Heidelberg: Springer Verlag, \tt arXiv:1209.5213\rm, 2012.
\bibitem{Bl/Br/Th2} D. Blackwell, L. Breiman, and A. J.
Thomasian, The capacities of a certain channel classes under random
coding, Ann. Math. Statist. Vol. 31, No. 3, 558-567, 1960.
\bibitem{Bl/Ca}V. Blinovsky and M. Cai, Classical-quantum arbitrarily varying wiretap
channel, Information Theory, Combinatorics, and Search Theory: In Memory of Rudolf Ahlswede, Lecture Notes in Computer Science, H. Aydinian, F. Cicalese, and C. Deppe
eds., LNCS Vol.7777,  234-246, Heidelberg: Springer Verlag, \tt{arXiv:1208.1151}\rm, 2012.
\bibitem{Bl/La}  M. Bloch and J. N. Laneman, On the secrecy capacity of arbitrary wiretap channels, Communication, Control, and Computing, Forty-Sixth Annual Allerton Conference
Allerton House, UIUC,  USA, 818-825, 2008.
\bibitem{Bo/Ca/Ca/De} H. Boche,  M. Cai, N. Cai,  and C. Deppe,  
 Secrecy capacities of compound quantum wiretap channels and applications, 
Physical Review A, Vol.89, No.5, 052320,
\tt arXiv:1302.3412\rm, 2014.
\bibitem{Bo/No} H. Boche and J. N\"otzel, Arbitrarily small amounts of correlation
for arbitrarily varying quantum channel,  J. Math. Phys., Vol. 54, Issue 11,  \tt arXiv 1301.6063\rm,
2013.
\bibitem{Bo/Wy} H. Boche and R. F.  Schaefer (Wyrembelski), Capacity results and super-activation for
wiretap channels with active wiretappers, IEEE Trans.  on Information Forensics and Security,   Vol. 8, No. 9, 1482-1496, 2013.
\bibitem{Ca/Wi/Ye} N. Cai, A. Winter, and R. W. Yeung, Quantum privacy and
quantum wiretap channels, Problems of Information Transmission, Vol.
40, No. 4,  318-336, 2004.
\bibitem{Cs/Ko}I.  Csisz\'ar and  J. K\"orner,  Information Theory: Coding Theorems for Discrete Memoryless Systems, Academic Press/Akademiai Kiao, 1981.
\bibitem{Cs/Na} I.  Csisz\'ar and  P. Narayan, The capacity of the arbitrarily varying channel revisited: positivity, constraints,
IEEE Trans. Inform. Theory, Vol. 34, No. 2,  181-193, 1988.
\bibitem{De} I. Devetak, The private classical information capacity and
quantum information capacity of a quantum channel, IEEE Trans. Inform. Theory, Vol. 51, No. 1, 44-55,  2005.
%\bibitem{De/Win} I. Devetak and A. Winter, Distillation of secret key and entanglement
%from quantum states, Proc. R. Soc. A, Vol. 461, 207-235, 2005.
%\bibitem{DiVi/Sh/Sm} D. P. DiVincenzo, P. W. Shor, and J. A. Smolin, Quantum channel
%capacity of very noisy channels, Phys. Rev. A, Vol. 57,  830,
%\tt arXiv:quant-ph/9706061\rm, 1998.
\bibitem{Rei}T. Ericson, Exponential error bounds for random codes in the arbitrarily varying channel,
IEEE Trans. Inform. Theory, Vol. 31, No. 1, 42-48, 1985.
\bibitem{Fa/Ka} K. Fazel and S. Kaiser, Multi-Carrier and Spread Spectrum Systems.
 From OFDM and MC-CDMA to LTE and WiMAX, 2. edition, ISBN 978-0-470-99821-2,
John Wiley \& Sons, New York, 2008.
\bibitem{Gi/Wo} G. Giedke and M. M. Wolf, Quantum communication: super-activated channels,
 Nature Photonics, Vol. 5, No. 10,
578-580, 2011.
%\bibitem{Ha} M. B. Hastings, A counterexample to additivity of minimum output
%entropy, Nature Physics, Vol. 5, 255, \tt arXiv:0809.3972 [quant-ph]\rm, 2009.
\bibitem{he-khisti-yener} X. He, A. Khisti, and A. Yener, Mimo multiple access channel with an arbitrarily varying eavesdropper: 
Secrecy degrees of freedom, IEEE Trans. Inform. Theory, Vol. 59, No. 8, 4733-4745, 2013.
\bibitem{Ho} A. S. Holevo,  The capacity of quantum channel with general signal states,
 IEEE Trans. Inform. Theory,  Vol. 44, 269-273, 1998.
\bibitem{ITW2010} W. Kang and N. Liu, Wiretap channel with shared key, IEEE Inf. Theory Workshop - ITW 2010 Dublin, 2010.
\bibitem{Li/Kr/Po/Sh} Y. Liang, G. Kramer, H. Poor, and S. Shamai, Compound wiretap channels,
EURASIP Journal on Wireless Communications and Networking - Special issue on wireless physical 
layer security  archive,  Vol. 2009,  Article No. 5,  2009.
\bibitem{Li/Wi/Zou/Guo} K. Li, A. Winter, X. B. Zou, G. C.  Guo, Private capacity of quantum channels is not additive,
Physical Review Letters, Vol. 103, No. 12, 120501, 2009.
%\bibitem{Llo} S. Lloyd, Capacity of the noisy quantum channel, Physical Review A,
%Vol. 55, No. 3, 1613-1622,  1997.
%\bibitem{Mi/Sch}  V. D. Milman and G. Schechtman, Asymptotic Theory of Finite Dimensional Normed Spaces.
%Lecture Notes in Mathematics 1200, Springer-Verlag, corrected second
%printing, Berlin, UK, 2001.
\bibitem{Wi/No/Bo2}J. N\"otzel,  M. Wiese,  and H. Boche, 
The Arbitrarily Varying Wiretap Channel --- Secret Randomness, 
Stability and Super-Activation, \tt arXiv:1501.07439\rm, 2015.
%\bibitem{Og/Na} T. Ogawa and H. Nagaoka, Making good codes for
%classical-quantum channel coding via quantum hypothesis testing,
%IEEE Trans. Inform. Theory, Vol. 53, No. 6, 2261-2266, 2007.
\bibitem{Opp} J. Oppenheim, For quantum information, two wrongs can make a
right, Science, Vol. 321, 1783, 2008.
%\bibitem{Pa} V. Paulsen, Completely Bounded Maps and Operator Algebras,
%Cambridge Studies in Advanced Mathematics 78, Cambridge University
%Press, Cambridge, UK, 2002.
\bibitem{Sch/Wes} B. Schumacher  and M. D. Westmoreland, Sending classical information via noisy quantum channels,
Phys. Rev., Vol. 56, 131-138, 1997.
%\bibitem{Sho} P. W. Shor, The quantum channel capacity and coherent information,
% Lecture Notes, MSRI Workshop on Quantum Computation, 2002.
\bibitem{Smi/Smo/Yar} G. Smith, J. A. Smolin,   and J. Yard, Quantum communication with Gaussian channels of zero quantum capacity,
Nature Photonics. Vol.  5, 624-627, 2011.
\bibitem{Smi/Yar} G. Smith and J. Yard, Quantum communication with zero-capacity channels,
Science Magazine,  Vol. 321, No. 5897, 1812-1815, 2008.
\bibitem{Wi/No/Bo} M. Wiese, J. N\"otzel, and H. Boche, A channel under simultaneous 
jamming and eavesdropping attack---correlated random coding capacities
 under strong secrecy criteria, accepted for publication in  IEEE Trans. Inform. Theory,
\tt arXiv:1410.8078\rm, 2014.
\bibitem{Wil}M. Wilde,   Quantum Information Theory,
 Cambridge University Press,  2013.
\bibitem{winter-locking} A. Winter, Weak locking capacity of quantum channels can be much larger than private capacity, 
Journal of Cryptography, Vol. 30, 1432-1378,
\tt arXiv:1403.6361\rm, 2016.
\bibitem{Wyn} A. D. Wyner, The wire-tap channel, Bell System Technical
Journal, Vol. 54, No. 8, 1355-1387, 1975.
\end{thebibliography}
\end{document}